%% file: main-ccs2026.tex
\def\ExtendedVersion{}
\ifdefined\ExtendedVersion
  \documentclass[sigconf,nonacm]{acmart}
\else
  \documentclass[sigconf]{acmart} 
\fi
\newif\ifextended
\ifdefined\ExtendedVersion \extendedtrue \else \extendedfalse \fi

\usepackage[english]{babel}
\usepackage{booktabs}
\usepackage{listings}
\usepackage{stmaryrd}    
\usepackage{mathtools}   
\usepackage[usenames,dvipsnames]{xcolor}
\usepackage{graphicx}
\usepackage{multirow}
\usepackage[inline]{enumitem}
\usepackage{amsmath,amsthm,mathtools,thmtools}
\usepackage[inference]{semantic}
\usepackage{cleveref}

\input{macro}

\allowdisplaybreaks

\AtBeginDocument{%
  }

\ifextended
\else
\copyrightyear{2026}
\acmYear{2026}
\setcopyright{cc}
\setcctype{by}
\acmConference[CCS '26] {Proceedings of the 2026 ACM SIGSAC Conference on Computer and Communications Security}{November 15--19, 2026}{The Hague, Netherlands.}
\acmBooktitle{Proceedings of the 2026 ACM SIGSAC Conference on Computer and Communications Security (CCS '26), November 15--19, 2026, The Hague, Netherlands}
\acmISBN{979-8-4007-2871-6/2026/11}
\acmDOI{10.1145/3830454.3846689}

\fi

\begin{document}

\title{Enforcement of In-Kernel Stateful Security Policies via eBPF}
\ifextended
  \subtitle{Extended Version}
  \titlenote{This is the extended version of the paper appearing in
    \emph{Proc.\ ACM SIGSAC Conference on Computer and Communications
    Security (CCS~'26)}.  It includes the appendices with the full formal development
    and the additional language constructs omitted from the proceedings version.}
\fi

\author{Letterio Galletta}
\orcid{0000-0003-0351-9169}
\affiliation{%
  \institution{IMT School for Advanced Studies Lucca}
  \city{Lucca}
  \country{Italy}
}
\email{letterio.galletta@imtlucca.it}

\renewcommand{\shortauthors}{Letterio Galletta}

\begin{abstract}
Many attacks against workloads running on multi-tenant systems are multi-step and history-dependent: sequences of innocuous operations whose malicious nature emerges only over an execution trace. 
Defending against them requires security policies that are \emph{stateful}, are enforced \emph{within the kernel}, and have a \emph{precise semantics}.  
Currently deployed proposals fail on at least one count: 
kernel's built-in syscall filtering and classical MAC frameworks are stateless, 
while current eBPF-based tools express their policies through ad hoc YAML rules that
cannot capture temporal relations among events, and whose semantics is defined only by the implementation.

We present \thetool, an in-kernel runtime-verification framework that satisfies the properties above.  
\thetool provides a policy language with a formal semantics that can express temporal relations among events. 
It also provides a type system that statically distinguishes events the kernel can \emph{control} from those it can only \emph{observe}.  
Every well-typed policy is compiled into a finite-state monitor proved
correct with respect to its semantics, and then into eBPF programs that run inside the kernel. 
We evaluate \thetool on seven case studies drawn from real-world attack patterns, and on a set of micro- and macro-benchmarks to show that the enforcement overhead remains compatible with production deployment.
\end{abstract}

\ifextended\else


\fi

\maketitle

\section{Introduction} 

\input{sections/introduction}

\section{Background}\label{sec:background}
\input{sections/background}

\section{\thetool in a Nutshell}\label{sec:overview}
\input{sections/overview}

\section{Threat Model}\label{sec:threat}
\input{sections/threat-model}

\section{The Policy Language}\label{sec:lang}
\input{sections/policy-lang}

\section{Monitor Synthesis}\label{sec:synthesis}
\input{sections/monitor-synthesis}

\section{Compilation to eBPF}\label{sec:compilation}
\input{sections/compilation-ebpf}

\section{Implementation and Evaluation}\label{sec:impl}
\input{sections/implementation}

\section{Related Work}\label{sec:related}
\input{sections/related-work}

\section{Conclusion}\label{sec:concl}
\input{sections/conclusion}

\begin{acks}
We thank the anonymous reviewers for their comments and suggestions, which helped improve the paper.
This work has been partially supported by RDS PTR 25-27 CYBER 2.1 ``Progetto Cybersecurity'' WP3 LA 3.21 -- CUP: D63C24001060001.
\end{acks}

\bibliographystyle{ACM-Reference-Format}
\balance
\bibliography{main-ccs2026}

\appendix 

\crefalias{section}{appendix}
\crefalias{subsection}{appendix}
\crefalias{subsubsection}{appendix}
\crefname{appendix}{Appendix}{Appendices}
\Crefname{appendix}{Appendix}{Appendices}

\section{Generative AI Usage}
During the preparation of this work, the author used Anthropic's Claude Code to scaffold the structure of the paper, revise its text, and develop the prototype implementation and the scripts needed for the experimental evaluations in \Cref{sec:impl} and for packaging the artifact.
After using this tool, the author reviewed and edited the content of the paper as needed, and manually inspected, validated and tested the generated code.
The author takes full responsibility for the final content of the publication.

\ifextended

\section{Formal Development}\label{sec:formal}
\input{sections/appendix-formal.tex}

\section{Additional Language Constructs}\label{app:extensions}
\input{sections/appendix-impl.tex}

\else

\fi

\end{document}

%% file: macro.tex
\usepackage{xspace}

\ifextended
  
\else
  
\fi

\lstdefinelanguage{bpfence}{
	morekeywords={schema,hook,params,ctx,let,happened,when,policy,apply,to,action,
		deny,kill,alert,forbid,then,within,import,as,use,and,or,not,any},
	morecomment=[l]{\#},
	morestring=[b]",
	sensitive=true
}
\newcommand{\thetool}{\textsc{BPFence}\xspace}
\newcommand{\field}{\ensuremath{\mathsf{field}}}
\newcommand{\Vars}{\ensuremath{\mathcal{V}}}
\newcommand{\Vals}{\ensuremath{\mathsf{Val}}}
\newcommand{\Id}{\ensuremath{\mathit{Id}}}

\newcommand{\Fbound}{\ensuremath{F^{\leq\Delta}}}
\newcommand{\Fboundj}{\ensuremath{F^{\leq\Delta_j}}}
\newcommand{\sat}{\ensuremath{\mathsf{sat}}}
\newcommand{\bnd}{\ensuremath{\mathsf{bnd}}}
\newcommand{\use}{\ensuremath{\mathsf{use}}}
\newcommand{\rulename}[1]{\textsc{#1}}

\ifextended
  \newcommand{\extref}[1]{\Cref{#1}}
  \newcommand{\extrefprop}[1]{\Cref{#1}}
\else
  \newcommand{\extref}[1]{the extended version~\cite{bpfence-extended}}
  \newcommand{\extrefprop}[1]{see the extended version~\cite{bpfence-extended}}
\fi
\newcommand{\mlet}{\texttt{let}\xspace}
\newcommand{\happened}{\texttt{happened}\xspace}
\newcommand{\when}{\texttt{when}\xspace}
\newcommand{\forbid}{\texttt{forbid}\xspace}
\newcommand{\within}{\texttt{then within}\xspace}
\newcommand{\Alert}{\textsc{Alert}\xspace}
\newcommand{\Deny}{\textsc{Deny}\xspace}
\newcommand{\Kill}{\textsc{Kill}\xspace}

\newcommand{\B}{\mathbb{B}_3}              
\newcommand{\ok}{\mathsf{ok}}
\newcommand{\pending}{\mathsf{pending}}
\newcommand{\violated}{\mathsf{violated}}
\newcommand{\dsl}[1]{\llbracket #1 \rrbracket_{\mathsf{DSL}}}
\newcommand{\sem}[1]{\llbracket #1 \rrbracket}
\newcommand{\pred}[1]{\llbracket #1 \rrbracket^a}
\newcommand{\Allowed}{\mathsf{Allowed}}

\newcommand{\GenMon}{\mathsf{GenMon}}
\newcommand{\NF}{\mathsf{NF}}
\newcommand{\updH}{\mathsf{UpdH}}
\newcommand{\pol}{\ensuremath{\Pi}}
\newcommand{\mtrue}{\ensuremath{\mathit{true}}}
\newcommand{\mfalse}{\ensuremath{\mathit{false}}}
\newcommand{\mnoapp}{\ensuremath{\mathit{na}}}
\newcommand{\Mem}{\ensuremath{\mathit{Mem}}}
\newcommand{\bad}{\ensuremath{\mathit{bad}}}
\newcommand{\pastF}{\ensuremath{\Phi_<}}
\newcommand{\hm}{\ensuremath{\mathcal{H}}}

\newcommand{\jm}{\ensuremath{\mathcal{J}}}
\newcommand{\eval}[3]{\ensuremath{\mathcal{E}\llbracket #1 \rrbracket^{#2}(#3)}}

%% file: sections/introduction.tex

Cloud-native deployments rely on a single shared kernel to host hundreds of mutually distrustful workloads: containers of different tenants, 
microservices managed by an orchestrator, and short-lived sandboxes spawned by build pipelines.
In this setting, the kernel is a critical isolation boundary, and a single compromised process can become a foothold for container escape, lateral movement across tenants, or persistent access to host resources.  
Operators rely on \emph{security policies} to constrain which operations a workload may perform and keep that boundary effective.  
However, vulnerabilities such as Dirty~Pipe~\cite{cve20220847}
witness how attackers can exploit the system calls allowed by default container profiles to escalate from a
compromised application to host privileges by invoking
operations that look innocuous in isolation.
What distinguishes such attacks is that the malicious behavior is \emph{multi-step} and
\emph{history-dependent}: a write to a sensitive file followed by an
\texttt{exec} syscall, or a successful invocation of \texttt{ptrace} followed by a shell spawn.
Containing these attack chains or simply detecting them reliably requires security policies that talk about \emph{execution traces}, not just about individual syscalls or system events.

In this setting, effective security policies must satisfy \emph{three}
complementary requirements.
First, they must be \emph{stateful}:
they need to remember the relevant history of the monitored
entity, e.g., a process or a container, and react to sequences of events and time windows.
Second, the policy language must have a \emph{well-defined semantics}:
the operator must be able to say what a policy denotes independently of its implementation,
and the toolchain must be able to reject policies it cannot enforce, rather than silently accepting them and producing a monitor whose behavior diverges from the operator's intent.
Third, enforcement must happen \emph{inside the kernel}: a
userspace monitor sits on the wrong side of the syscall boundary,
pays a context switch on every relevant event, and is exposed to
time-of-check/time-of-use races.
In-kernel enforcement, by
contrast, mediates the operation synchronously, on the same
execution path that performs it, and observes exactly the
arguments the kernel will use.
Unfortunately, the mechanisms practitioners actually deploy fail to meet at least one of these three requirements.

On the one hand, system-call filters such as
\texttt{seccomp-bpf}~\cite{seccomp-bpf} run entirely in the kernel,
but their policies are stateless: a filter cannot remember that the same process
has just opened \texttt{/etc/shadow}.  
Most production seccomp profiles, therefore, are simple static allow lists.
Mandatory access control frameworks such as
SELinux~\cite{smalley2001selinux} and AppArmor~\cite{apparmor} are
more expressive, but their decisions are still stateless with
respect to recent history and their policies are notoriously hard to
author and audit.
On the other hand, modern eBPF-based runtime security tools such as Falco~\cite{falco2023} and Tetragon~\cite{tetragon2024}
leverage programmable in-kernel storage and can react to security-relevant events, providing the mechanisms for stateful monitoring.
However, they miss the second requirement:
their policies are expressed through ad hoc YAML rules whose semantics is defined only by the implementation: there is no
specification of \emph{what} a rule denotes, no static check that a
rule asks the kernel to enforce something the kernel can in fact
enforce.
In addition, multi-step and bounded-time properties must be encoded as collections of
correlated rules whose joint behavior is left to the operator's intuition.

These previous efforts show that the in-kernel machinery to enforce programmable policy is already available, but what is still missing is a
policy layer that meets \emph{all} requirements at
once: expressive enough to talk about execution traces and time windows, 
equipped with a semantics the toolchain and the operator can reason about, and enforced inside the kernel.

\paragraph{Our proposal}
This paper presents \thetool, an in-kernel runtime-verification framework that meets the three requirements above.  
\thetool offers a policy language with a formal trace semantics whose constructs capture recurring classes of stateful security properties: safety obligations, bounded-response liveness, historical predicates, sliding-window counts, and cross-event correlation.  
A type system distinguishes events the kernel can actually \emph{control} (deny or terminate before they take effect) from those it can only \emph{observe}, and statically rejects policies that ask the kernel to enforce something it cannot.  
Each well-typed policy is compiled to a monitor that is provably correct with respect to its  semantics, and then to eBPF programs that run inside the kernel.  
While related ideas have appeared separately in prior work,
\thetool is the first framework to combine them under the eBPF verifier's constraints:
the policy language is designed to trade expressiveness for verifier-compatibility, so that
the resulting monitors not only pass the eBPF verifier checks but also run with an acceptable overhead, as our evaluation confirms.

Concretely, this paper makes the following contributions:
\begin{itemize}[nosep,leftmargin=*]
  \item \emph{A policy language for in-kernel stateful security properties}
    (\Cref{sec:lang}) with a denotational trace semantics and a type system that statically separates controllable from observable events, ruling out policies the kernel cannot enforce.
  \item \emph{A monitor synthesis procedure} (\Cref{sec:synthesis}) that
    turns every well-typed policy into a finite-state monitor. 
    We prove that this monitor is correct with respect to the policy language semantics.
  \item \emph{A compilation schema}
    (\Cref{sec:compilation}) that turns synthesized monitors into
    eBPF programs accepted by the in-kernel verifier and enforced
    synchronously on controllable events.
  \item \emph{A PoC implementation}
    (\Cref{sec:impl}) including the \thetool{} compiler and runtime, and an essential standard library of LSM hook schemas and policy templates. 
    We evaluate the effectiveness of our implementation on seven case studies drawn from real-world attack patterns; compare the expressiveness of our policy language against some upstream Falco/Tetragon rules; and measure the monitoring overhead via micro- and macro-benchmarks.
\end{itemize}

Before presenting these contributions, \Cref{sec:background} reviews the necessary background.  
\Cref{sec:overview} outlines \thetool through an example, and
\Cref{sec:threat} states our threat model.  
After presenting the contributions, \Cref{sec:related} situates \thetool with respect to the relevant literature, and \Cref{sec:concl} draws some conclusions and discusses future work.
All the auxiliary notions required by our formal development, the proofs of our theorems, and further details on our implementation are in \extref{sec:formal,app:extensions}.
The artifact to reproduce our experiments is available online~\cite{repo}.

%% file: sections/background.tex

Extended BPF (eBPF)~\cite{ebpf2023} is a technology that allows
user-defined programs to run inside the Linux kernel in response to
system events, without modifying kernel source code or loading kernel
modules.
An eBPF program is compiled to a platform-independent bytecode, loaded
into the kernel, JIT-compiled to native code, and attached to a
\emph{kernel hook}.

Kernel hooks intercept and modify the kernel execution flow at
predefined points.
They determine \emph{when} an eBPF program executes and \emph{what} it
can observe or control.
Three hook families are relevant to security monitoring:
\begin{enumerate*}[label=(\roman*)]
\item \emph{LSM hooks} belong to the Linux Security Module
  framework~\cite{wright2002lsm} and execute the eBPF program at
  security decision points \emph{before} a sensitive operation
  completes: if the program returns an error code, the operation is
  blocked synchronously;
\item \emph{tracepoints}~\cite{trace-kernel} are static probe points
  scattered throughout the kernel that trigger a program \emph{after} an event has
  occurred, so they can only observe, not block, a syscall;
\item \emph{kprobes}~\cite{kprobe-kernel} are dynamic probes that can
  be attached to any kernel routine to collect debugging and
  performance information non-disruptively.
\end{enumerate*}
Only LSM hooks make an event \emph{controllable}, so we can deny the action or kill the offending process;
tracepoints and kprobes make it only \emph{observable}, so we can only raise an alert.
This distinction led us to design the type system of \Cref{sec:lang} that
statically rejects policies attempting to deny an observable-only event.

At load time, every eBPF program is statically verified by the
in-kernel \emph{BPF verifier}, which enforces memory safety and
guarantees termination: programs with uninitialized registers, unbounded
loops, out-of-bounds accesses, or excessive instruction complexity are
rejected.
These constraints ensure that eBPF programs cannot crash or hang the
kernel, but they restrict the class of expressible computations,
a limitation that affects the design of the policy language of \Cref{sec:lang} and the compilation strategy of \Cref{sec:compilation}.

eBPF programs communicate with one another and with user space through
\emph{BPF maps}: key--value data structures that reside in kernel
memory and persist across program invocations.
There are several map types; the ones most relevant to us
are hash maps, arrays, and ring buffers. 
Maps are shared, so two eBPF programs attached to different hooks can
interact using the same map.

%% file: sections/overview.tex

Consider an attacker who gains initial access to a remote host, e.g., via a web
vulnerability, and attempts an SSH lateral movement: read an SSH private key (\texttt{/home/\allowbreak user/\allowbreak .ssh/id\_rsa}), open a connection to port~22 on an internal host, and execute a command on the remote machine.
Each individual action is legitimate: many services read SSH keys, connect on port~22, or spawn shells or subprocesses. 
Only the full three-step sequence reveals the attack~\cite{mitreT1021:004}.

To detect such an attack sequence and enforce a response to block the offending operations,
we define the following \thetool policy:
\begin{lstlisting}[numbers=left, numberstyle=\tiny,stepnumber=1,numbersep=5pt]
import stdlib linux files        // read syscall
import stdlib linux network      // connect syscall 
import stdlib linux process      // exec syscall

let key_read = happened(
  read("/root/.ssh/*") or read("/home/*/.ssh/*")
)

let ssh_connected = happened(
  connect(_, 22)
) when key_read

policy lateral_movement {
  apply to pid  action kill
  forbid exec(_) when ssh_connected
}
\end{lstlisting}
The policy involves three events raised by invoking three syscalls: 
\texttt{read} (file open and read), \texttt{connect}, and \texttt{exec} (process execution).
Each event is bound to an LSM hook via an event schema definition provided by the standard
library (e.g., \texttt{read} maps to the LSM hook
\texttt{security\_\allowbreak file\_open}, which intercepts both open and
read operations).
The actual parameters of each syscall are automatically extracted at runtime and made available  as event fields in the policy (e.g., the file path for \texttt{read}, the destination port for \texttt{connect}).

The \texttt{let} declaration at line 5 introduces the \emph{history
predicate} \texttt{key\_read}.
This predicate becomes true when the
monitored process reads the SSH private keys of the admin or a user from their \texttt{.ssh} directory.
Once this happens, the predicate remains true in the future: the \texttt{happened} operator records whether a predicate has ever been satisfied.
The predicate \texttt{ssh\_connected} at line 9 is \emph{conditioned} on the previous one: 
it becomes true only when a connection to port~22 is attempted and \texttt{key\_read} is already true.
An isolated SSH connection or an isolated key read does not affect it.
Note that in the two definitions the syntax \texttt{*} matches zero or more characters within a single directory level, and \texttt{\_} denotes ``any address''.

The policy \texttt{lateral\_movement} is defined at line 13.
Its header specifies that the \emph{enforcement scope} is at the process level 
(\texttt{apply to pid}), namely one monitor instance per process. 
The body of the policy contains a single \texttt{forbid} clause guarded by \texttt{ssh\_connected}: when the predicate holds, any subsequent \texttt{exec} call is blocked and the process is terminated (\texttt{action} \texttt{kill} in the policy header).
For this to be possible in practice, the \texttt{exec} event in the \texttt{forbid} clause must be bound to an LSM hook that can deny the corresponding operation (controllable event).
Our type system statically ensures the enforceability of the policy: it assigns the violation type $C$ (controllable) to the event \texttt{exec} according to its hook and checks that the \texttt{kill} action is indeed compatible with the type.
Had  \texttt{exec} been bound to a tracepoint (observable only), the compiler would have rejected the policy due to a type mismatch.

For the lateral movement policy, the compiler generates three BPF programs, one per hook, that share their state through a common eBPF map.
A user-space daemon loads the generated BPF bytecode, sets up the shared state, attaches the programs to the corresponding hooks, and starts dispatching violation events that  the loaded programs generate.
The next sections formalize each phase and present the implementation.

Finally, note that current eBPF-based monitoring tools cannot enforce policies such as the lateral movement above.
\texttt{seccomp-bpf}~\cite{seccomp-bpf} and AppArmor~\cite{apparmor} are \emph{stateless}: they can block \emph{all} connections to port~22 or \emph{none}, but cannot condition the decision on a previously observed file read.
Falco~\cite{falco2023} can only \emph{alert} on individual events; multi-step detection requires external engines.
Tetragon~\cite{tetragon2024} can block operations, but its filtering conditions can only inspect fields of the \emph{current} event: there is no mechanism to maintain state across events and condition a decision on what happened earlier.
Additional comparisons are in \Cref{sec:impl}.

%% file: sections/threat-model.tex

We consider a Linux host running user-space processes, possibly inside containers, monitored by \thetool.
We assume the \thetool loader installs the compiled monitors \emph{before} process execution begins.
The trusted computing base (TCB) comprises the Linux kernel (including the eBPF verifier and runtime), the \thetool compiler, its loader, and the hardware.

The attacker can execute arbitrary code in a monitored process and can issue any system call supported by the execution environment.
Attack strategies include multi-step privilege escalation, lateral movement, and data exfiltration.
We do not address kernel exploits, verifier bugs, or side-channel attacks.
We assume that the attacker cannot compromise any component of the TCB, and that they do not actively try to evade known \thetool policies (e.g., by splitting a kill chain across  
distinct processes via \texttt{fork}) or attack the monitor itself.

Under this threat model, we show that the generated monitors enjoy \emph{soundness} (no false positives) and \emph{detection completeness} (no missed violations), assuming that events are attached to the relevant hooks.
For controllable events, enforcement is \emph{synchronous}: the offending system call is denied before the kernel operation completes.
We also assume that updates to the monitor states are atomic and that policy reloads never expose a partially configured monitor, so guarantees hold throughout the lifetime of each monitored process; \Cref{sec:maps,sec:impl:tool} discuss how the implementation discharges these assumptions.

%% file: sections/policy-lang.tex
\thetool\ policies are written in a declarative domain-specific language (DSL) whose formal semantics is defined over a data-aware temporal logic. 
The logic operates on traces of events that abstractly represent the system calls invoked by a running program, and combines three key features: temporal operators from timed linear temporal logic with past~\cite{bauer2011}; a three-valued semantics over the domain $\mathbb{B}_3 = \{\ok, \pending, \violated\}$ to express monitor decisions; 
and data-aware event predicates that bind and constrain field values across events via correlation variables.
In addition, a type system statically classifies each policy by how violations arise at runtime and ensures that enforcement actions are associated only with policies whose violations can be effectively intercepted.

We present the trace model (\Cref{sec:events}), the temporal logic (\Cref{sec:logic}), the DSL (\Cref{sec:dsl}), and its type system (\Cref{sec:types}).

\subsection{Event Model}
\label{sec:events}
We abstractly represent a program execution as a trace of events generated by system call invocations.
According to the nature of events in eBPF, we distinguish between two types: events that can be intercepted via LSM hooks, enabling pre-event enforcement, and events that are only observable post hoc via tracepoints or kprobes.

Formally, let $\Sigma = \Sigma^C \uplus \Sigma^O$ be a finite set of event types,
partitioned into \emph{controllable} events $\Sigma^C$ and \emph{observable} events $\Sigma^O$.
An \emph{event} $e$ is a tuple $\langle E, \mathit{fields}, t\rangle$
where $E\in\Sigma$, $\mathit{fields}$ is a finite typed
parameter map, and $t$ is a monotone timestamp.
An \emph{execution trace} is a finite, time-ordered sequence of events
$\sigma = e_1 e_2 \cdots e_n$ with $e_i.t < e_{i+1}.t$ where we use the dot notation to access the fields of $e$.
In the following, we denote by $\sigma_k$ a trace of length $k$ and by $e_k$ its last event.

\subsection{Data-aware Temporal Logic}
\label{sec:logic}
Our data-aware temporal logic is defined at two levels: event predicates and temporal formulas, both of which use a three-valued semantics with an environment to handle variable binding.

\paragraph{Event predicates}
Intuitively, an event predicate is a propositional formula evaluated on the events of a trace. 
Since events carry data values, we introduce the notions of \emph{variables} and \emph{guards} to refer to these values and express constraints over them.

Formally, let $\Vars$ be a set of variables (ranged over by $X, Y, \ldots$).  
An \emph{event predicate} is defined as:
\begin{align*}
P \;::=\; & \alpha \mid \lnot P \mid P\land P \mid P\lor P \qquad \qquad \quad
\alpha \;::=\; (E, g) \\
g \;::=\; & 
\mathsf{true} \mid \field\sim\ell \mid \field={?X} \mid \field=X \mid g\land g	
\end{align*}
where an \emph{event atom} $\alpha = (E,g)$ pairs an event type $E$
with a guard $g$.
A guard $g$ is a conjunction of constraints over the fields of an event, where $\ell \in \Vals$ is a literal value, 
${\sim}\in\{=,{\neq},{<},{\leq},{>},{\geq}\}$ is a comparison operator, and $\field$ denotes the name of an event field. 
In addition, a guard introduces two variable constructs with distinct roles: 
$\mathsf{field}={?X}$ is a \emph{binding}, which captures the value of a field into a fresh variable $X$; whereas $\mathsf{field}=X$ is a \emph{constraint}, requiring the field to be equal to the previously bound variable $X$.
Note that  
the correspondence between kernel hooks and the events referenced in our predicates $P$ is assumed to be given by $\Sigma$. 

We assume hereafter that a well-formedness condition holds: each variable must be bound before use, and no bindings may appear under $\lnot$ or $\lor$ (see \extref{sec:well:form}).
Given an event $e$ and an \emph{environment} $\rho\colon\Vars \to \Vals$, we interpret an event predicate $P$ on a three-valued boolean domain $K_3 = \{\mtrue, \mfalse, \mnoapp \}$ using a strong Kleene semantics~\cite{Kleene52}.
The idea is that \emph{inapplicable} value $\mnoapp$ captures situations in which the event $e$ is not among those $P$ is talking about. 
Formally, the semantics of event predicate $P$ is given by $\pred{P}(e,\rho) = (v, \rho')$ where $v \in K_3$ and $\rho'$ is an updated environment with new bindings introduced by the guards of $P$.
The complete definition of $\pred{\cdot}$ is in \extref{sec:pred:sem}. 
Hereafter, we write $\pred{P}(\sigma_k,\rho)$ to denote the application of $\pred{\cdot}$ to the last event $e_k$ of a trace $\sigma_k$.

\begin{example}[Event predicates]\label{ex:pred}
Consider the event predicate at line 6 of the policy of \Cref{sec:overview}, written as
\begin{align*}
	P_k \;=\;& (\mathtt{read},\, \mathit{path} \sim \texttt{/root/.ssh/*}) \;\lor\\
	& (\mathtt{read},\, \mathit{path} \sim \texttt{/home/*/.ssh/*})
\end{align*}
where $\sim$ denotes glob matching, and the following three events:
\begin{align*}
	e_1 &= \langle\mathtt{read},\ \{\mathit{path}\mapsto\texttt{/home/alice/.ssh/id\_rsa}\},\ t_1\rangle\\
	e_2 &= \langle\mathtt{connect},\ \{\mathit{dest}\mapsto\texttt{10.0.0.7},\ \mathit{port}\mapsto 22\},\ t_2\rangle\\
	e_3 &= \langle\mathtt{read},\ \{\mathit{path}\mapsto\texttt{/etc/hosts}\},\ t_3\rangle.
\end{align*}
We have $\pred{P_k}(e_1, \emptyset) = (\mtrue, \emptyset)$ because the event arguments match those of the second disjunct. 
In contrast, $\pred{P_k}(e_2, \emptyset) = (\mnoapp, \emptyset)$: the
predicate $P_k$ does not apply to $e_2$ because it simply does not concern
\texttt{connect} events.  
Finally, $\pred{P_k}(e_3, \emptyset) = (\mfalse, \emptyset)$: the event type of $e_3$ matches, but the guard fails on both disjuncts.  Telling $\mfalse$ apart from $\mnoapp$ allows the monitor to evaluate predicates over unrelated event types
without spurious violations.
\end{example}

\paragraph{Temporal formulas}
The syntax of temporal formulas is:
\[
\phi \;::=\; P \mid \neg \phi \mid \phi\land\phi \mid {<}\phi \mid G\,\phi \mid P_1\,\Fbound\,P_2
\]
where $P$ is an event predicate, ${<}$ is the \emph{once} operator, $G$ is the \emph{always} operator, and
$\Fbound$ is a \emph{bounded eventually} operator where $\Delta \in \mathbb{N}$ is a finite temporal bound: $P_1\,\Fbound\,P_2$
requires that every time $P_1$ holds, $P_2$ holds within $\Delta$ time units.
We assume that a well-formedness condition holds to ensure each variable is bound before use and that no binding appears under $\neg$ or ${<}$ (see \extref{sec:well:form}).

The temporal formulas are evaluated over the semantic domain $\B = \{\ok, \pending, \violated\}$.
The value $\ok$ indicates that the monitor has not detected any violation so far, while $\violated$ indicates that a bad prefix has been detected.
The value $\pending$ denotes that the truth of the formula cannot yet be decided due to open obligations: based on the information currently available, the monitor requires additional future events to reach a verdict.

Formally, given a trace $\sigma_k$ and an environment $\rho$, the semantics of a formula $\phi$ is given by $\sem{\phi}(\sigma_k,\rho) = (v, \rho')$ where $v \in \B$ and $\rho'$ is an updated environment.
The evaluation function $\sem{\cdot}$ is inductively defined over the syntax of formulas. Below, we illustrate the semantics of the once and always operators through an example and report the definition of $\sem{\cdot}$ for the case of bounded response; the other cases are in \extref{sec:app:log:sem}.

\begin{example}[Once and always]\label{ex:once}
Consider the predicate $P_k$ of
\Cref{ex:pred} and the trace $\sigma_3 = e_1e_2e_3$ built from the events taken from that example. Also, let $P_e$ be the single-atom predicate $(\mathtt{exec},\, \mathsf{true})$.
We have that $\sem{{<}P_k}(\sigma_3,\emptyset) = \ok$ since $e_1$ reads a
key; in contrast, $\sem{{<}P_e}(\sigma_3,\emptyset) = \pending$ because no
$\mathtt{exec}$ has occurred yet and a future one could still make it $\ok$.
Finally, $\sem{G(\lnot P_k)}(\sigma_3,\emptyset) = \violated$, since $P_k$
holds on the prefix $e_1$, making $\lnot P_k$ $\violated$ there, and no
extension can retract it.
Both temporal operators look at all the prefixes of $\sigma_3$, but do so differently:
the once operator asks whether some prefix succeeded, whereas the always operator asks whether some prefix failed.
\end{example}

For \emph{bounded response} formulas $P_1 \Fbound P_2$, the $\pending$ value represents an obligation that is still unresolved: if the obligation is not fulfilled within   
the specified time bound, a violation occurs. 
We define the \emph{active instance set} generated every time $P_1$ holds:
\begin{align*}
\mathcal{I}(\sigma_k,\rho) & \;\stackrel{\mathit{def}}{=}\; \\
& \bigl\{\,(i,\rho_i,d_i) \;\big|\;
i<k,\;
\pred{P_1}(\sigma_i,\rho)=(\mtrue,\rho_i),\;
d_i=e_i.t+\Delta
\,\bigr\}
\end{align*}
where every instance $(i,\rho_i,d_i)$ records the index $i$ of the event where $P_1$ was true, the produced environment $\rho_i$, and the absolute deadline $d_i$.
An instance $(i,\rho_i,d_i)$ is satisfied on $\sigma_k$ if:
\begin{align*}
\sat(i, \rho_i, d_i, \sigma_k) &
\;\stackrel{\mathit{def}}{\iff}\; \\*
& \exists\, j \in (i, k].\;
e_j.t \leq d_i
\;\land\;
\pi_1\!\bigl(
\pred{P_2}(\sigma_j,\, \rho_i)
\bigr) = \mtrue
\end{align*}
Then, the semantics of the formula is:
\[
\sem{P_1\, \Fbound\, P_2}(\sigma_k,\, \rho)
\;=\; (v,\, \rho)
\]
where
\[
v =
\begin{cases}
	\violated & \exists(i,\rho_i,d_i)\in\mathcal{I}(\sigma_k, \rho).\;
	e_k.t > d_i \land \neg\sat(i,\rho_i,d_i,\sigma_k) \\
	\pending  & \exists(i,\rho_i,d_i)\in\mathcal{I}(\sigma_k, \rho).\;
	e_k.t\leq d_i \land \neg\sat(i,\rho_i,d_i,\sigma_k) \\
	\ok       & \text{otherwise}
\end{cases}
\]
The first case captures the situation in which the deadline has expired, and the expected event has never occurred; from that point on, the violation is definitive and cannot be remedied by any future events.
The second case captures the situation in which there is still time before the deadline, and thus the monitor can continue waiting. The third case captures the situation in which the formula becomes $\ok$ within the prescribed time bound.
Note that the $\pending$ value arises only when an obligation is active, but its deadline
has not yet expired, and that it is always transient since all deadlines are finite.

\begin{example}[Bounded response and bindings]\label{ex:bind}
Consider the following formula $\phi$, which requires every file to be
closed within five seconds of being opened:
\[
	\phi \;=\; G\bigl( (\mathtt{open},\, \mathit{ino}={?X})
	\;F^{\leq 5\mathrm{s}}\;
	(\mathtt{close},\, \mathit{ino}=X) \bigr)
\]
The $\mathtt{open}$ atom binds the i-node field to the variable
$X$, and the $\mathtt{close}$ atom then constrains its own i-node to match
$X$. In this way, the two events are correlated by the value of the i-node.
Upon the event $e_4 = \langle\mathtt{open},\{\mathit{ino}\mapsto 8123\},t_4\rangle$ the
$\mathtt{open}$ atom yields $(\mtrue,\{X\mapsto 8123\})$, so the instance
$(4,\{X\mapsto 8123\},t_4+5\mathrm{s})$ enters $\mathcal{I}$, and the
$\mathtt{close}$ atom is evaluated under the environment it records.  
Only a close on i-node $8123$ satisfies the instance, making $\sat$ hold and the verdict $\ok$. 
Any event other than a close of i-node $8123$ leaves the verdict $\pending$;
once the deadline expires, it turns $\violated$.
Note that the correlation between related $\mathtt{open}$ and $\mathtt{close}$ events would not be possible without $X$.
\end{example}

\subsection{The \thetool DSL}
\label{sec:dsl}

The \thetool DSL is a user-oriented declarative language for expressing security policies. 
It provides structured constructs for common security patterns, including \emph{safety}, \emph{bounded response}, \emph{history-based}, and \emph{taint} policies.
The DSL is intentionally designed to ensure online monitorability and runtime enforceability and to effectively target eBPF.
The DSL has no independent semantics, but each policy denotes a formula of the temporal logic defined above: the DSL is thus a monitorable fragment of the logic.
Note that since $\sem{\pol}$ is a function of a single trace, \thetool can express and monitor only trace properties: hyperproperties such as non-interference are outside the scope of this paper.

\paragraph{Syntax}
Given a finite set of event types $\Sigma$, the abstract syntax of a \thetool program is:
\[
\pol \;::=\; \langle d^*\; p, A \rangle 
\]
where $d^*$ is a (possibly empty) sequence of historical-predicate definitions, $p$ is the policy body, and $A \in \{\Alert, \Deny, \Kill\}$ is the action to take when a violation occurs:
\begin{align*}
	d \;::=\; & \mlet\ H := \happened(P) [\when\ H'] \\
	p \;::=\; & \forbid\ P\ [\when\ H] \mid\ \when\ P_1\ \within\ \Delta\ P_2 \mid\ p_1 \land p_2
\end{align*}
where $P, P_1, P_2$ are event predicates introduced above, $H, H'\in\mathit{Id}$ are historical-predicate identifiers, and $\Delta \in \mathbb{N}$ is a finite time bound.

A historical-predicate definition $d$ introduces a predicate that captures a property depending on the past history of events. A declaration introduces a new predicate named $H$ and may include an optional \when\ clause, conditioning the definition of $H$ on another historical predicate $H'$.
Historical predicates can be referenced in the policy body $p$.
A $\forbid$ clause specifies an event predicate $P$ that must never hold; it may optionally include a $\when$ clause, restricting the prohibition to situations in which the historical predicate $H$ holds.
A $\when\ \dots\ \within$ clause expresses a bounded response property: whenever $P_1$ holds, $P_2$ must hold within the time bound $\Delta$.
Finally, policy clauses can be combined via conjunction.

\paragraph{Semantics}\label{sec:denot}
The semantics is given by first translating the policy into a temporal formula and then interpreting this formula on a trace $\sigma_k$.
More precisely, we proceed as follows.
The historical-predicate declarations $d^*$ evaluate to a definition environment
$D$ mapping each identifier $H$ to a logical formula. 
We have $D(H) = {<}P$ for $\mlet\ H:=\happened(P)$, and
$D(H) = {<}(D(H')\land P)$ for the \when\ clause.
Acyclicity of $d^*$ (enforced statically by the type system) ensures $D$ is well-defined.

Then, given a definition environment $D$ and a policy body $p$, we define the \emph{function}
$\dsl{p}^D$ that maps policies to temporal formulas, inductively on the syntax of $p$:
\begin{align*}
	\dsl{\texttt{forbid}\ P}^D &= G(\lnot P) \\
	\dsl{\texttt{forbid}\ P\ \texttt{when}\ H}^D
	&= G\bigl(\lnot(D(H)\land P)\bigr) \\
	\dsl{\texttt{when}\ P_1\ \texttt{then within}\ \Delta\ P_2}^D
	&= G\bigl(P_1\,F^{\leq\Delta}\,P_2\bigr) \\
	\dsl{p_1 \land p_2}^D &= \dsl{p_1}^D\land\dsl{p_2}^D
\end{align*}
A $\forbid$ clause is translated into a $G$ formula, while a $\when\ \dots\ \mathtt{then}$ $\mathtt{within}$ clause is translated into an always bounded-response formula.
From the definition above, it follows directly that every policy is translated into a conjunction of $G$ formulas.

\begin{example}[From policy to formula]\label{ex:denot}
Consider the \texttt{lateral\_\allowbreak movement} policy of
\Cref{sec:overview}, whose body $p$ is the clause
\forbid $P_e$ \when $\mathtt{ssh\_connected}$, where $P_e$ is as in \Cref{ex:once}.  
The two let declarations yield an environment $D$ with the bindings
\begin{align*}
	D(\mathtt{key\_read}) &= {<}P_k\\
	D(\mathtt{ssh\_connected}) &= {<}({<}P_k \land P_c)
\end{align*}
where $P_c = (\mathtt{connect},\, \mathit{port}=22)$ and $P_k$ is as in
\Cref{ex:pred}, making it evident that the second declaration is conditioned on the first.  By the
$\forbid$--$\when$ case of $\dsl{\cdot}^D$ we obtain
\begin{align*}
	\dsl{p}^D &= G\bigl(\lnot(D(\mathtt{ssh\_connected}) \land P_e)\bigr)\\
	&= G\bigl(\lnot({<}({<}P_k \land P_c) \land P_e)\bigr)
\end{align*}
as the denotation of our policy. Note that the nesting of ${<}$ ensures that the inner ${<}P_k$ already holds when the connection
occurs, and that both hold when the $\mathtt{exec}$ is attempted.  
\end{example}

Finally, given a \thetool program $\pol = \langle d^*\;p, A\rangle$ and a trace $\sigma_k$, the semantics of $\pol$ is defined as:
\[
	\sem{\pol}(\sigma_k)
\;\triangleq\;
\pi_1\!\bigl(
\sem{\phi}(\sigma_k,\;\emptyset)
\bigr)
\;\in\;\B 
\qquad 
\text{where } \phi=\dsl{p}^{D(d^*)}. 
\]
We first translate $\pol$ into a temporal formula $\phi$, evaluate $\phi$ on the trace $\sigma_k$ and an empty environment $\emptyset$, and take the resulting value in $\B$: $\ok$ means no
clause is violated and no obligation is pending; $\pending$ means at
least one $F^{\leq\Delta}$ obligation is active but its deadline has not
expired; $\violated$ means a safety clause is violated
\emph{now}, or a deadline has elapsed without the required event.
Crucially, $\violated$ is irrevocable: no future extension of the trace
can undo it.
When a violation occurs, the monitor performs action $A$.

\paragraph{Formal results}
The following two results characterize the monitorability of our language.
The first establishes that once a DSL program detects a violation, no future event can undo the decision:
\begin{restatable}[Violation Persistence]{theorem}{persi}\label{thm:persistence}
	For every DSL program $\pol$ and finite traces $\sigma_k$, $\sigma_j$ such that $\sigma_k$ is a prefix of $\sigma_j$,
	if $\sem{\pol}(\sigma_k) = \violated$ then $\sem{\pol}(\sigma_j) = \violated$.
\end{restatable}
The second result ensures that each violating trace has a first violating prefix:
\begin{restatable}[Existence of the first violating prefix]{proposition}{vioprefix}\label{prop:fvp}
	Let $\pol$ be a DSL program and $\sigma_k$ be a finite trace such
	that $\sem{\pol}(\sigma_k) = \violated$. There exists a unique
	prefix $\sigma_{i^*}$ of $\sigma_k$ such that
	\[
	i^* = \min\{\, i \leq k \mid \sem{\pol}(\sigma_i) =
	\violated \,\}.
	\]
	We call $\sigma_{i^*}$ the \emph{first violating prefix} of $\sigma_k$.
\end{restatable}

\subsection{Type System}\label{sec:types}

The semantics defines when an execution trace violates a policy but does not guarantee that a policy can be effectively enforced at runtime. 
In particular, a semantically valid policy may still be: 
\emph{not enforceable}, because a violation is triggered by an observable event that can only be detected after its occurrence; or
\emph{not directly enforceable}, because the violation arises from a deadline
expiration rather than from a specific controllable event.
To address these issues, we introduce a type system that classifies DSL programs according to how violations may arise at runtime: due to a controllable event (type $C$), an observable event (type $O$), or by a deadline expiration (type $T$). 
The type system statically verifies that the action $A$ associated with the program is compatible with its violation type, e.g., \Deny is permitted only for policies of type $C$.
Since types describe modes of violation and enforcement capabilities, we refer to them as \emph{violation types}.
In addition to the typing rules presented here, our implementation also verifies standard well-formedness: every identifier must be declared before use, historical-predicate definitions must be acyclic, and values used in event-predicate guards must match the types declared in the corresponding event schemas.

\paragraph{Types and type environment}
The set of violation types consists of three elements
$
\tau \in \mathcal{T} = \{C,\, O,\, T\}
$
where $C$ means that the violation is triggered by a controllable event; $O$ that the violation is caused by an observable event; and $T$ that the violation arises from a deadline expiration or the absence of an event.
The types are ordered to form a finite lattice  $C \sqsubseteq O \sqsubseteq T$:
$C \sqsubseteq O$ because a controllable event is a special case of an observable one, and $O \sqsubseteq T$ because a
violation anchored to a specific event is strictly more enforceable than one determined by a timeout. We denote by $\sqcup$ the join operator.

The classification of events as controllable or observable depends on the deployment environment.
We capture this dependency through a typing environment $\Gamma$, which assigns to each event type $E \in \Sigma$ its corresponding violation type:
$\Gamma(E) = C$ if the event $E$ can be intercepted via an LSM hook, and $\Gamma(E) = O$ if $E$ is obtained through a tracepoint or kprobe hook.

Since policy bodies may refer to historical predicates $H$, 
we introduce an environment $\Theta\colon \Id \to \mathcal{T}$ that records the violation type of each identifier. The environment $\Theta$ is constructed inductively from the historical-predicate definitions $d^*$ occurring in the program, as defined by the rules in \extref{sec:app:typing}.

\begin{figure*}[t]
\begin{gather*}
\inference[\rulename{PEvent}]{ }{\Gamma \vdash (E, g) : \Gamma(E)} 
\qquad 
\inference[\rulename{PAnd}]{\Gamma \vdash P_1 : \tau_1 \quad \Gamma \vdash P_2 : \tau_2}{\Gamma \vdash P_1 \land P_2 : \tau_1 \sqcup \tau_2}
\qquad
\inference[\rulename{Comp}]{\Gamma;\Theta \vdash p_1 : \tau_1 \quad \Gamma;\Theta \vdash p_2 : \tau_2}{\Gamma;\Theta \vdash p_1 \land p_2 : \tau_1 \sqcup \tau_2}
\qquad
\inference[\rulename{Forbid}]{ \Gamma \vdash P : \tau }{ \Gamma;\Theta \vdash \forbid\ P : \tau }
\\[1ex]
\inference[\rulename{ForbidWhen}]{\Gamma \vdash P : \tau \quad \tau_H = \Theta(H)}{\Gamma;\Theta \vdash \forbid\ P\ \when\ H : \tau \sqcup \tau_H}
\qquad
\inference[\rulename{Response}]{ \Gamma \vdash P_1 : \tau_1 \quad \Gamma \vdash P_2 : \tau_2}{\Gamma;\Theta \vdash \when\ P_1\ \within\ \Delta\ P_2 : T}
\qquad
\inference[\rulename{Prog}]{\Gamma;\emptyset \vdash d^* : \Theta \quad \Gamma; \Theta \vdash p : \tau \\ \Allowed(\tau, A)}{\Gamma \vdash \langle d^*\, p,\; A \rangle : \checkmark}		
\end{gather*}
\caption{Relevant typing rules for \thetool DSL.}\label{fig:typing}	
\Description{Seven inference rules, in two rows, assign violation types to event predicates, under the environment Gamma, and to policy bodies, under Gamma and Theta. Top row: PEvent gives an event atom, made of an event type E and a guard g, the type that Gamma assigns to E; PAnd and Comp give a conjunction, of two event predicates and of two policy bodies respectively, the join of the two types; Forbid gives a forbid clause the type of its event predicate P. Bottom row: ForbidWhen gives a forbid clause guarded by a historical predicate H the join of the type of P and of the type that Theta records for H; Response gives a when-within clause the type T, whatever its two predicates are typed; Prog concludes that a whole program is well typed if the declarations derive Theta, the body has type tau, and the predicate Allowed holds of tau and of the action A.}
\end{figure*}

\paragraph{Typing rules}
The most relevant typing rules are in \Cref{fig:typing} (the full set is in \extref{sec:app:typing}).
The rules for event predicates, characterized by the judgment $\Gamma \vdash P : \tau$, assign a violation type $\tau$ to an event predicate $P$ under the environment $\Gamma$.
The type for an event $\alpha=(E,g)$ is taken from the environment $\Gamma$ (rule \rulename{PEvent}). 
Rule \rulename{PAnd} joins the types of the two sub-terms to ensure that a violation can occur in the worst possible way.

The rules for policies have typing judgments of the form $\Gamma; \Theta \vdash p : \tau$.
Rule \rulename{Comp} takes the join of the types of the two sub-terms, reflecting that a violation may arise from either of the conjuncts.
In rule \rulename{Forbid}, the resulting type is that of the event predicate $P$, reflecting the fact that \texttt{forbid} denotes a safety property whose violations are always anchored to a specific event. 
Rule \rulename{ForbidWhen} is similar, but the type is $\tau_H \sqcup \tau_P$, where $\tau_P$ is the type of $P$ and $\tau_H$ is the type of the historical predicate $H$.
In the previous two rules, the join means that a violation occurs when both $D(H)$ and $P$ hold simultaneously.
Since both event predicates and historical predicates are typed in $\{C, O\}$, in either case the resulting type is at most $O$, never $T$: no timeout can cause a violation.
In rule \rulename{Response}, the premises $\Gamma \vdash P_1 : \tau_1$ and $\Gamma \vdash P_2 : \tau_2$ ensure that $P_1$ and $P_2$ are well-typed event predicates, but their types do not affect the resulting type $T$, which is determined solely by the temporal structure $\Fbound$: 
any policy with a bounded-response obligation may produce a timeout violation, regardless of whether the involved events are controllable or observable.

Rule \rulename{Prog} derives the environment $\Theta$ from the declarations $d^*$, obtains the violation type $\tau$ of the policy body $p$, and checks that the action $A$ is compatible with $\tau$ via the predicate $\Allowed$:
\[
\Allowed(\tau, A) \;\triangleq\;
\begin{cases}
	A \in \{\Alert, \Deny, \Kill\}
	& \text{if } \tau = C, \\
	A = \Alert
	& \text{if } \tau \in \{O, T\}.
\end{cases}
\]
The actions \Deny and \Kill are permitted only when $\tau = C$, 
while the action \Alert is always permitted, since the monitor can always raise a warning.
Note that the type $\tau$ of a policy depends on the hooks available in the deployment environment: the same event may have type $C$ where an LSM hook is available and type $O$          
otherwise, making the policy \Deny-able in the first case and only \Alert-able in the second.

\begin{example}[Typing a policy]\label{ex:typing}
Consider again the \texttt{lateral\_\allowbreak movement} policy of \Cref{sec:overview}, on a
deployment where \texttt{read}, \texttt{connect}, and \texttt{exec} are all bound to LSM
hooks, so that $\Gamma$ assigns $C$ to all of them.  The let declarations
yield an environment $\Theta$ such that
$\Theta(\mathtt{key\_read}) = \Theta(\mathtt{ssh\_connected}) = C$.  
To type the policy, we apply rule \rulename{ForbidWhen}, which joins the type of $\mathtt{ssh\_connected}$ with that of \texttt{exec}, resulting in $\Gamma;\Theta \vdash p : C$ for its body $p$. 
Since $\Allowed(C, \Kill)$ holds, \rulename{Prog} accepts the program with the \texttt{kill} action.  Consider extending the policy with a new clause denoting the formula of \Cref{ex:bind},
which requires every file to be closed within five seconds of being opened.  
Even when $\Gamma$ assigns $C$ to both \texttt{open} and \texttt{close},
rule \rulename{Response} assigns the type $T$ to the new clause because the violation may come from the expiry of the deadline. 
By rule \rulename{Comp}, the type propagates to the whole program, so, for the $\Allowed$ predicate to hold, the action can only be \Alert.
\end{example}

\paragraph{Type soundness}
The following theorem establishes the soundness of the type system with respect to the semantics: the violation type $\tau$
assigned to a policy is a sound upper bound on the controllability class of the event that triggers the violation at runtime, when such an event exists:
\begin{restatable}[Type Soundness]{theorem}{typecor}\label{thm:type:cor}
Let $\Gamma$ be a type environment, $\pol = \langle d^*\,p, A\rangle$
be a DSL program such that $\Gamma;\emptyset \vdash d^* : \Theta$ and
$\Gamma;\Theta \vdash p : \tau$.
For each finite trace $\sigma_k$ such that $\sem{\pol}(\sigma_k) = \violated$, let
$\sigma_{i^*}$ be the first violating prefix of $\sigma_k$.
Then:
\begin{itemize}
	\item if $\tau \in \{C, O\}$, the violation is triggered by
	an offending event $e_{i^*}$ such that
	$\Gamma(e_{i^*}.\mathit{type}) \sqsubseteq \tau$;
	\item if $\tau = T$, either the violation is triggered by an
	offending event $e_{i^*}$ such that
	$\Gamma(e_{i^*}.\mathit{type}) \sqsubseteq O$, or it is due to a
	deadline expiration: there exists a bounded-response sub-formula
	$P_1 \Fbound P_2$ of $\dsl{p}^D$ and an instance $(i, \rho_i, d_i)
	\in \mathcal{I}(\sigma_{i^*}, \emptyset)$ of its active instance
	set such that $e_{i^*}.t > d_i$ and $\neg\,\sat(i, \rho_i, d_i,
	\sigma_{i^*})$.
\end{itemize}
\end{restatable}

%% file: sections/monitor-synthesis.tex
A DSL program is transformed into a \emph{monitor automaton} that performs the monitoring and the enforcement.
We first introduce a \emph{normal form} to capture a structural property of the denotations of DSL programs (\Cref{sec:nf}).
We then define the synthesis function $\GenMon$, which translates such normal form into a monitor, and establish that the resulting monitor is sound and complete with respect to the semantics of \thetool DSL (\Cref{sec:automata}).

\subsection{Normal Form}
\label{sec:nf}

Every formula produced by the denotation function $\dsl{\cdot}$ of \Cref{sec:lang} either has the $G$ operator at the outermost level or is a conjunction of two formulas. 
Since the operator $G$ distributes over $\land$ (\extrefprop{prop:g-conj}),
we can always collect all conjuncts under a single $G$, obtaining the following normal form:
\[
G \Bigl(
\bigwedge_{i \in I} \lnot\psi_i
\;\land\;
\bigwedge_{j \in J} \bigl(P^j_1\,\Fboundj\,P^j_2\bigr)
\Bigr)
\]
where each $\psi_i$ is either an event predicate $P$ or a
formula of the form ${<}\phi_H \land P$, whereas $P^j_1, P^j_2$ are event predicates,
$\Delta_j \in \mathbb{N}$ is a finite bound, and $I$, $J$ are finite index
sets.
The formula ${<}\phi_H$ is the one associated by $D$ with the historical predicate $H$.
The rationale is that the normal form highlights two kinds of constraints a policy can express. 
The conjuncts $\lnot\psi_i$ express safety conditions: a violation occurs when $\psi_i$ holds on the current event, either because the event predicate $P$ is satisfied directly
(\forbid\ $P$), or because both the past condition
${<}\phi_H$ and the predicate $P$ hold simultaneously
(\forbid\ $P$ \when $H$). 
The conjuncts $P^j_1 \Fboundj P^j_2$ express bounded-response obligations that are violated by a deadline expiry.

Given a DSL program $\pol = \langle d^* p, A\rangle$, we define the normalization
function $\NF(\pol)$ inductively on the structure of $p$ (see \extref{sec:app:norm:form}). 
Intuitively, for a \forbid\ clause, it produces a safety conjunct $\lnot\psi_i$; for a \within clause, it produces a
bounded-response conjunct $P^j_1 \Fbound P^j_2$.
For a conjunction $p_1 \land p_2$, it applies $\NF$
recursively to $p_1$ and $p_2$, and merges the resulting formulas into a
single $G$.
The normalization is correct with respect to the semantics of formulas, processes each clause exactly once, and terminates in a number of steps equal to the number of clauses in $p$ (\extrefprop{th:nf-sem}).

\begin{example}[Normal form of a two-clause policy]\label{ex:nf}
Consider the \texttt{lateral\_\allowbreak movement} policy of \Cref{sec:overview}
extended, as in \Cref{ex:typing}, with a bounded-response clause
that requires every file to be closed within five seconds of
being opened. Denote by $P_o = (\mathtt{open},\, \mathit{ino}={?X})$ and by
$P_{cl} = (\mathtt{close},\, \mathit{ino}=X)$ the two predicates of \Cref{ex:bind}.
The body of the extended policy is the conjunction $p = p_1 \land p_2$ of the two clauses, so its denotation is $\dsl{p}^D = \dsl{p_1}^D \land \dsl{p_2}^D$:
\begin{align*}
	\dsl{p_1}^D &= G\bigl(\lnot({<}({<}P_k \land P_c) \land P_e)\bigr)\\
	\dsl{p_2}^D &= G\bigl(P_o\;F^{\leq 5\mathrm{s}}\;P_{cl}\bigr)
\end{align*}
where the denotation of $p_1$ is the one computed in \Cref{ex:denot}.
Each clause of $p$ thus carries its own $G$. 
The normalization function $\NF$ collects the two into a single one:
\[
\NF(\pol) \;=\; G\bigl(
\lnot({<}({<}P_k \land P_c) \land P_e)
\;\land\;
P_o \; F^{\leq 5\mathrm{s}} \; P_{cl}
\bigr)
\]
\end{example}

\subsection{Building Monitor Automata}
\label{sec:automata}

The normal form plays a key role in monitor synthesis: by the semantics of $G$, every $\violated$ outcome of $\NF(\pol)$ on a prefix $\sigma_k$ is due to a violation
of at least one conjunct, either a safety conjunct
$\lnot\psi_i$  or a bounded-response conjunct $P^j_1
\Fboundj P^j_2$. 
This separation allows us to construct the monitor compositionally: we
synthesize a local automaton for each conjunct and then combine them into a single monitor by parallel composition.

Specifically, a \emph{monitor automaton} is a tuple 
\[
\mathcal{A} = \langle Q,\,q_0,\,\Mem,\,{\rightarrow},\,\bad\rangle
\]
where $Q$ is a finite control-state set, $q_0$ is the initial state,
$\Mem$ is the monitor memory (defined below), ${\rightarrow}$ is the transition relation, and
$\bad\subseteq Q\times\Mem$ is the set of bad configurations.
Below, we first define the monitor memory and then the synthesis function $\GenMon$ to translate each conjunct into a local monitor.

\paragraph{Monitor memory}
The memory consists of two components. The first records whether each past formula ${<}\phi_H$ occurring in the safety conjuncts of $\NF(\pol)$ has been satisfied at some point in the past (\emph{historical memory}).
The second records the set of active obligations and a violation flag for each bounded-response conjunct.

Formally, let $\pastF$ denote the set of past formulas ${<}\phi_H$ occurring in $\NF(\pol)$. 
The \emph{historical memory} is a function 
$\hm\colon \pastF \to \{\ok, \pending\}$
where $\hm({<}\phi_H) = \ok$ means that ${<}\phi_H$ has been satisfied in the past, whereas $\hm({<}\phi_H) = \pending$ means that it is has not. 
The historical memory is monotone: once an entry becomes $\ok$, it remains $\ok$ on all subsequent steps of the monitor.

For each bounded-response conjunct $P^j_1 \Fboundj P^j_2$, the monitor maintains a \emph{clause state} $(I_j, f_j)$. 
The set $I_j$ collects the currently \emph{active instances}: each instance $(d, \rho)$ records the absolute deadline $d$ and the environment $\rho$ captured when $P^j_1$ matched.
The flag $f_j \in \{\ok, \violated\}$ is set to $\violated$ when at least one instance expires without the required response, and remains $\ok$ otherwise.

The \emph{monitor memory} is then a pair $m = (\hm,
\jm)$, where $\hm$ is the historical memory and $\jm$ maps each bounded-response conjunct $j \in J$ to a clause state $(I_j, f_j)$.
We denote by $m_\hm$ and $m_\jm$ the first and the second element of $m$, respectively.

\begin{figure*}[t]
\centering
\begin{gather*}
\inference[\rulename{S-Viol}]{
	\eval{\psi_i}{m_\hm}{e} = \mtrue
}{
	(\ok, m) \xrightarrow{\,e\,}^s_i (\violated, m)
}	
\qquad 
\inference[\rulename{S-Ok}]{
	\eval{\psi_i}{m_\hm}{e} \neq \mtrue
}{
	(\ok, m) \xrightarrow{\,e\,}^s_i (\ok, m)
}
\qquad
\inference[\rulename{S-Absorb}]{
	\phantom{X}
}{
	(\violated, m) \xrightarrow{\,e\,}^s_i (\violated, m)
}
\\[1ex]
\inference[\rulename{R-Ok}]{
	U(m_\jm(j), e) = (I', \ok)
	\quad
	m'_\jm = m_\jm[j \mapsto (I', \ok)]
}{
	(\ok, m) \xrightarrow{\,e\,}^r_j (\ok, m')
}
\qquad
\inference[\rulename{R-Viol}]{
	U(m_\jm(j), e) = (I', \violated)
	\quad
	m'_\jm = m_\jm[j \mapsto (I', \violated)]
}{
	(\ok, m) \xrightarrow{\,e\,}^r_j (\violated, m')
}
\\[1ex]
\inference[\rulename{R-Absorb}]{}{
	(\violated, m) \xrightarrow{\,e\,}^r_j (\violated, m)
}
\end{gather*}
\caption{Rules defining the transition relations $\to^s_i$ and $\to^r_j$ of the local automata $A^s_i$ and $A^r_j$, for every $i \in I$ and $j \in J$.}\label{fig:monitor:trans}
\Description{Six inference rules, in three rows, define the transition relations of the safety and of the bounded-response local automata; each rule reads one event and rewrites a configuration, that is a control state paired with the monitor memory. The top row holds the safety rules: S-Viol moves the automaton from ok to violated when the conjunct psi-i evaluates to true on the event under the historical memory, S-Ok keeps it in ok otherwise, and neither touches the memory. The lower rows hold the bounded-response rules, which apply the update function U to the clause state of the conjunct j and to the event: according to the flag U returns, R-Ok keeps the automaton in ok and R-Viol moves it to violated, and both store the updated clause state. Each family has a further rule, S-Absorb and R-Absorb, with no premise: an automaton already in violated stays there, so that state is absorbing.}
\end{figure*}

\paragraph{Safety automaton}
For each safety conjunct $\lnot\psi_i$ in $\NF(\pol)$, the synthesis function $\GenMon$ returns a safety monitor automaton $A^s_i = \langle Q^s, q^s_0, \Mem, \to^s_i, \bad^s
\rangle$ where:
\[
Q^s = \{\ok, \violated\}, \quad
q^s_0 = \ok, \quad
\bad^s = \{(\violated, m) \mid m \in \Mem\}.
\]
Note that the violation condition depends solely on the control state:
$A^s_i$ is in violation if and only if its state is $\violated$, regardless of the memory contents.
The transition relation $\to^s_i$ is defined by the rules \rulename{S-*} at the top of \Cref{fig:monitor:trans}.
In the rules, we use an \emph{evaluation function} 
$\eval{\psi_i}{\hm}{e} \in K_3$ that evaluates a safety conjunct $\psi_i$ on a single event $e$ using the current historical memory $\hm$ to resolve past sub-formulas. 
For $\psi_i = P$, evaluation reduces to the evaluation of the predicate $P$ on $e$.
For $\psi_i = {<}\phi_H \land P$, we first resolve the past condition by looking it up in $\hm$: if $\hm({<} \phi_H) = \pending$, the evaluation results in  $\mfalse$ regardless of $P$; otherwise the evaluation results in the evaluation of $P$ on $e$. 
Rule \rulename{S-Viol} fires when $\psi_i$ evaluates to $\mtrue$ on the event $e$ under the current historical memory $m_\hm$: in this case, the automaton transitions to the $\violated$ state.
Rule \rulename{S-Ok} covers the cases  $\mfalse$ and $\mnoapp$, both of which leave the automaton in the $\ok$ state: $\mnoapp$ means the current event is irrelevant to $\psi_i$ and does not constitute a violation. 
Rule \rulename{S-Absorb} makes $\violated$ an absorbing state:
once a safety violation is detected, it is irrevocable, mirroring the semantics of $G(\lnot\psi_i)$.

\paragraph{Bounded-response automaton}
For each bounded-response conjunct $P^j_1\,\Fboundj
\,P^j_2$, the synthesis function $\GenMon$ constructs a bounded-response automaton $A^r_j = \langle Q^r, q^r_0, \Mem,\to^r_j, \bad^r \rangle$ where:
\[
Q^r = \{\ok, \violated\}, \quad
q^r_0 = \ok, \quad
\bad^r = \{(\violated, m) \mid m \in \Mem\}.
\]
The automaton tracks the clause state $(I_j, f_j)$ in the shared memory.
For each event $e$, the automaton: 
\begin{enumerate*}[label*=(\emph{\roman*})]
	\item checks whether any active instance has expired without satisfaction;
	\item removes from $I_j$ all instances that are expired or
	satisfied by $e$; and
	\item adds a new instance if $P^j_1$ is triggered by $e$, recording the deadline $e.t + \Delta_j$ and the environment $\rho'$ captured at
	activation.
\end{enumerate*}
These steps are performed by the update function $U((I_j, f_j), e)$.

The transition relation $\to^r_j$ is defined by the rules \rulename{R-*} in \Cref{fig:monitor:trans}, where $m'$ is the memory $m$ updated with the new
clause state.
Rule \rulename{R-Ok} fires when the memory update produces the violation flag $\ok$: \emph{no} active instance has expired without satisfaction, so the automaton remains in $\ok$ and the
clause state is updated to $(I', \ok)$, reflecting the removal of satisfied instances and the possible addition of a new one. 
Rule \rulename{R-Viol} fires when the update
produces the violation flag $\violated$: \emph{some} active instance has expired without satisfaction, so the automaton transitions to the $\violated$ state and the clause state is updated to $(I', \violated)$. 
Rule \rulename{R-Absorb} makes the $\violated$ state absorbing: once a deadline violation is detected, the memory is left unchanged and the violation is irrevocable.

\paragraph{Monitor automaton}
Given a normal form formula $\NF(\pol)$, the synthesis function $\GenMon$ constructs the global monitor as the parallel composition of one local automaton for each conjunct:
\[
\mathcal{A} \;=\;
\Bigl(\|_{i \in I}\,A^s_i\Bigr)
\;\|\;
\Bigl(\|_{j \in J}\,A^r_j\Bigr)
\]
where the set of control states is $Q = \prod_{i \in I} Q^s \times \prod_{j \in J} Q^r$, the initial state $q_0$ is given componentwise by $(q_0)_i = q^s_0$ for every
$i \in I$ and $(q_0)_j = q^r_0$ for every $j \in J$, the initial monitor memory is $m =(\hm,\jm)$ such that $\hm$ maps every entry to $\pending$ and $\jm$ maps every $j$ to the clause state $(\emptyset, \ok)$.
A global configuration $(q, m)$ is \emph{bad} if and
only if at least one local automaton is in $\violated$:
\[
\bad = \{(q, m) \mid \exists\, i \in I.\; q_i = \violated \;\lor\;
\exists\, j \in J.\; q_j = \violated\}.
\]

On each event $e$, the monitor transition performs two
steps. 
First, the  historical memory $\hm$ is updated to $\hm'$: 
for each past formula ${<}\phi_H$, the corresponding entry is set to $\ok$ if $\phi_H$ holds on $e$, and left unchanged otherwise. 
This update is performed before any local automaton fires, so that all automata read the same updated $\hm'$. 
Second, all local automata fire independently: for every $i \in I$ the automaton $A^s_i$ performs a $\to^s_i$ transition, and for every $j \in J$ the automaton $A^r_j$ performs a $\to^r_j$ transition, on the same updated memory. 
Since safety automata do not modify any clause state and the clause states of distinct bounded-response automata are disjoint, all updates are independent. 
The formal definition of the global monitor behavior is in \extref{sec:app:glob:mon}.

After each event $e$, the monitor applies the enforcement function $\mathsf{Enf}$ to the updated configuration $(q', m')$:
\[
\mathsf{Enf}(q', m') =
\begin{cases}
	\mathit{allow} &
	\text{if } (q', m') \notin \bad \\
	A & \text{otherwise } 
\end{cases}
\]
where $A \in \{\Alert, \Deny, \Kill\}$ is the declared action in the DSL program. 
Recall that for a well-typed program, the type system guarantees statically that $A$ respects the nature of the event $e$.

\ifextended
Note that the resulting monitor is deterministic (\Cref{prop:determinism}) and has size linear in the size of the formula and historical-predicate declarations (\Cref{prop:monitor-size}).
\else
Note that the resulting monitor is deterministic and has size linear in the size of the formula and historical-predicate declarations; both properties are proved in the extended version~\cite{bpfence-extended}.
\fi 
These two properties are critical to achieve an effective eBPF implementation.

\paragraph{Correctness}
The following theorems ensure the correctness of the monitor synthesis function $\GenMon$. 
\emph{Soundness} guarantees that the monitor never raises false alarms, whereas  \emph{detection completeness} guarantees that every semantic violation is
eventually detected. 
Together, they ensure that the monitor is faithful to the formal semantics of $\NF(\pol)$:
\begin{restatable}[Soundness]{theorem}{monsound}\label{thm:msound}
For every DSL program $\pol$ and finite trace $\sigma_k$, if $\mathcal{A} =
\GenMon(\NF(\pol))$ reaches a bad configuration after
reading $\sigma_k$, then $\sem{\pol}(\sigma_k) = \violated$.
\end{restatable}
\begin{restatable}[Completeness]{theorem}{moncompl}\label{thm:mcompl}
For every DSL program $\pol$ and finite trace $\sigma_k$, if $\sem{\pol}
(\sigma_k) = \violated$, then $\mathcal{A} = \GenMon(\NF(\pol))$ reaches a bad configuration on some prefix $\sigma_{k'}$ with $k' \leq k$, provided that \emph{(i)} every
event relevant to the policy is delivered to the monitor and \emph{(ii)} at least one event is delivered after every deadline expiration.	
\end{restatable}

Note that detection completeness requires two additional assumptions: first, that every event relevant to the policy is delivered to the monitor without loss; second, that for bounded-response clauses, at least one event is delivered after every deadline expiration, so that the monitor has the opportunity to detect the timeout. 
The first assumption holds when the eBPF hook is attached before any relevant event can occur; the second is satisfied in practice by the continuous stream of kernel events.
%

%% file: sections/compilation-ebpf.tex

This section describes how a monitor $\mathcal{A}$ produced in \Cref{sec:synthesis} is compiled into concrete eBPF programs that enforce the policy inside the Linux kernel.
We give an overview of the compilation scheme (\Cref{sec:comp:overview}), show how the monitor state is realized via BPF maps (\Cref{sec:maps}), present the per-event processing logic (\Cref{sec:processing}), and then describe the runtime support that completes the enforcement loop (\Cref{sec:runtime}).

\subsection{Compilation Overview}
\label{sec:comp:overview}

Each eBPF program is attached to exactly one kernel hook and can only intercept events raised at that hook.
A policy, however, can refer to events from multiple hooks: a bounded-response clause $P^j_1\,\Fboundj\,P^j_2$ may trigger on one hook and expect a response on another, and a safety conjunct $\lnot({<}\phi_H \land P)$ may record a historical flag on one hook and check it on another.
The \thetool compiler therefore translates a policy $\pol$ into one BPF program per hook mentioned in the policy.
Since these programs collectively implement the single monitor $\mathcal{A}$, they must share the monitor memory $(\hm, \jm)$ so that the state written by one program is visible to the others.

Concretely, the compilation of a program $\pol$ is driven by its normal form $\NF(\pol)$ (\Cref{sec:nf}), which splits the policy into safety conjuncts $\lnot\psi_i$ and bounded-response conjuncts $P^j_1\,\Fboundj\,P^j_2$.
Each conjunct maps directly to a block of generated code: safety conjuncts become \texttt{if}--\texttt{goto} checks implementing the rules \rulename{S-Viol}/\rulename{S-Ok} of \Cref{fig:monitor:trans}, while bounded-response conjuncts become lookups and updates on the monitor state implementing the $U$ function and the rules \rulename{R-Viol}/\rulename{R-Ok}.
Since the number of clauses is fixed, all conjuncts are \emph{unrolled at compile time}: the generated code is a straight-line sequence with no loops, which guarantees that every program passes the eBPF verifier's termination check.

Additionally, the generated programs run in the kernel and capture events from all processes simultaneously, making the monitor inherently system-wide.
However, policies track the behavior of individual entities independently, e.g., a taint flag set by one process must not propagate to another.
A policy defines the \emph{unit of monitoring} by declaring a \emph{scope} (\texttt{pid}, \texttt{tgid}, \texttt{cgroup}, or \texttt{namespace}): two events belong to the same monitored entity if they share the same scope value.
To achieve this grouping, on each intercepted event, the generated code extracts an \emph{execution key} $\kappa$ using an appropriate BPF helper (e.g., \texttt{bpf\_get\_current\_pid\_tgid()} for per-process scope, \texttt{bpf\_get\_current\_cgroup\_id()} for per-container scope).
The execution key $\kappa$ partitions the monitor memory: each monitored entity maintains its own independent copy $(\hm_\kappa, \jm_\kappa)$.
However, kernel identifiers can be recycled, for example, once a process exits, its \texttt{pid} can be reassigned to an unrelated process, which would then inherit the state accumulated under the same key.
For \texttt{pid}-scoped policies, the generated code therefore invalidates the historical memory of a key when the process it identifies exits, so that the monitor memory follows the lifetime of that process and not the reuse of its identifier.
Below, we detail how this per-entity state is concretely stored and shared across programs.

\subsection{Representing the Monitor State}
\label{sec:maps}

eBPF programs have no heap and no writable global variables, so the only mechanism for keeping persistent state is the \emph{BPF map}.
The generated \thetool programs use statically sized BPF maps to maintain the per-entity monitor memory $(\hm_\kappa, \jm_\kappa)$.

The historical memory $\hm_\kappa$ is implemented by a hash table keyed by a pair $(\kappa, h)$, where $h$ is a compile-time numeric identifier assigned to each historical predicate ${<}\phi_H$.
When the generated program evaluates $\phi_H$ to true, it sets the corresponding entry to~$1$; the entry is never reset while the monitored entity exists, preserving the monotonicity of $\hm$.

The active instance set $I_j$ of each bounded-response clause is stored in a map 
keyed by a pair $(\kappa, j)$, where each entry holds a deadline (a timestamp in nanoseconds) and the correlation environment $\rho'$ captured at activation of the instance, encoded as a fixed-size byte array.
When the map is full, an LRU policy discards the least-recently-used entry; as long as all the active instances fit in it, the generated program faithfully realizes the formal semantics of $I_j$ of \Cref{sec:automata}, which assumes an unbounded $I_j$.

The memory space required by a monitored entity can be determined by analyzing the policy: since the keys have the form $(\kappa, h)$ and $(\kappa, j)$, an entity occupies at most $|\pastF|$ entries of the first map and $|J|$ of the second.

Finally, the monitor is system-wide and its programs run on every core, so the state of one entity may be updated concurrently.
The generated code relies on the atomic operations the kernel provides for eBPF maps, and makes atomic those updates that read a value and write it back.

\subsection{Per-event Processing}
\label{sec:processing}

\Cref{fig:pseudocode} shows the pseudocode of a generated BPF program.
For each intercepted event $e$, the program first extracts the execution key~$\kappa$ from the event context and, if the policy restricts enforcement to a target set, short-circuits to \texttt{allow} for entities outside it (lines 3--5).
The code then reads the relevant event fields from the hook arguments into a local instance of a suitable structure (lines 7--9). 

\begin{figure}[t]
\begin{lstlisting}[
  language=C,
  basicstyle=\ttfamily\footnotesize,
  keywordstyle=\bfseries,
  commentstyle=\itshape\color{gray},
  frame=single,
  numbers=left,
  numberstyle=\tiny,
  xleftmargin=1.8em,
  escapeinside={(*@}{@*)},
  morekeywords={SEC,BPF_PROG}
]
SEC("lsm/hook_name")
int BPF_PROG(bpfence_hook, ...) {
  // 1. Key extraction and target filter
  u64 kappa = bpf_get_current_pid_tgid();
  if (!in_target_set(kappa)) return 0;

  // 2. Event extraction
  struct bpfence_ev ev;
  extract_fields(&ev, ctx);

  // 3. Historical-flag update (*@$\updH$@*)
  for each (*@${<}\phi_H \in \hm$@*):     // unrolled during compilation
    if ((*@$\phi_H$@*) holds on ev)
      history_map[kappa, h] = 1;

  // 4. Safety-clause checks (*@$\lnot\psi_i$@*)
  for each (*@$\psi_i \in S$@*):           // unrolled during compilation
    if ((*@$\sem{\psi_i}^{H_\kappa}(\mathtt{ev})$@*))
      goto violation;

  // 5. Bounded-response update (*@$P^j_1\,\Fboundj\,P^j_2$@*)
  for each clause j in R:        // unrolled during compilation
    rv = response_map[kappa, j];
    if (rv && expired(rv))       // (i) detect
      goto violation;
    if (rv && (*@$P^j_2$@*) matches ev)    // (ii) satisfy
      delete response_map[kappa, j];
    if ((*@$P^j_1$@*) matches ev)          // (iii) activate
      response_map[kappa, j] =
        { deadline: now + (*@$\Delta_j$@*), env: (*@$\rho'$@*) };

  return 0; // allow
violation:
  emit_alert(kappa, policy, clause, ...);
  // Deny: return -EPERM;
  // Kill: bpf_send_signal(SIGKILL);
}
\end{lstlisting}
\caption{Pseudocode of a generated BPF program.}\label{fig:pseudocode}
\Description{C-like pseudocode of the BPF program generated for one LSM hook, in five steps: extract the execution key and allow entities outside the target set; extract the event fields; set the historical flag of every past formula that holds; jump to the violation label if a safety conjunct holds; and, for each bounded-response clause, detect an expired deadline, delete a satisfied instance, activate a new one. The loops over clauses are unrolled at compile time. At the violation label the program emits an alert and, for the deny and kill actions, returns an error or signals the process.}
\end{figure}

The core of the program implements the monitor behavior as described in \Cref{sec:automata}.
The program first evaluates all historical predicates: for each ${<}\phi_H$, if $\phi_H$ holds on the current event, the corresponding entry in $\hm_\kappa$ is set to~$1$ (lines 11--14).
All historical updates are performed \emph{before} any safety or response check, so that subsequent evaluations read an updated $\hm_\kappa$.
Next, for each safety conjunct $\lnot\psi_i$, the program evaluates $\sem{\psi_i}^{\hm_\kappa}$ on the current event.
If the result is $\mtrue$, the event violates $\lnot\psi_i$, and the program jumps immediately to the enforcement handler (lines 16--19). 
Finally, for each bounded-response clause $P^j_1\,\Fboundj\,P^j_2$, the program updates the active instance set $I_j$ following a three-phase logic: it detects expired instances, removes satisfied ones, and activates new ones when $P^j_1$ matches (lines 21--30).

When a violation is detected, the BPF program applies the action declared in the policy:
\Deny returns \texttt{-EPERM} from the LSM hook, blocking the operation before it completes;
\Kill invokes \texttt{bpf\_send\_signal(SIGKILL)};
\Alert emits the violation event to user space without blocking.
A bounded-response violation occurs when a deadline expires without the required response.
Since the BPF program only executes when an event arrives, the runtime daemon provides a dedicated \emph{timeout scanner} that periodically checks for expired deadlines.

\subsection{Runtime Support}
\label{sec:runtime}

A user-space daemon loads the generated BPF programs and attaches them to their corresponding kernel hooks.
When a policy spans multiple hooks, the programs must share the monitor memory $(\hm_\kappa, \jm_\kappa)$ introduced in \Cref{sec:maps}.
The daemon achieves this sharing by loading programs sequentially: after the first program is loaded, the daemon binds each subsequent program's map slots to the maps already created by the first, so that all programs in the policy read and write the same BPF maps.
This mechanism enables cross-hook correlation of variables and values.

Once all programs are attached, the daemon performs two tasks.
First, it polls a shared ring buffer where the BPF programs write violation events (each carrying the policy identifier, execution key~$\kappa$, violated clause, and contextual information) and dispatches them to a configurable handler.
Second, it runs a \emph{timeout scanner} that periodically checks the response map for expired deadlines, detecting bounded-response violations.

%% file: sections/implementation.tex

\thetool has been implemented as two open-source tools available online~\cite{repo}: a compiler that translates policies into loadable eBPF programs, and a daemon that loads the eBPF programs, attaches them to their corresponding kernel hooks, and manages the violation notification.
\Cref{sec:impl:tool} briefly presents our implementation, whereas 
\Cref{sec:eval} presents the experimental evaluation.

\subsection{Implementation}
\label{sec:impl:tool}

The implementation consists of an OCaml compiler (${\sim}7K$ LoC) and a C
runtime daemon (${\sim}2K$ LoC).
The compiler translates a \texttt{.bpfence} source file into C files, which are then compiled to \texttt{.bpf.o} objects via \texttt{clang} and loaded by the daemon. 
In addition to the tasks described in \Cref{sec:runtime}, the daemon dispatches violation events and manages policy updates.
Violation events are processed by a configurable handler chain that ships with
three built-in handlers: 
\texttt{log} (structured logging to file or stdout), \texttt{alert}
(forwarding to syslog), and \texttt{exec}
(execution of a shell command to handle the event).
Policy updates can be applied at runtime, without interrupting
enforcement, upon receipt of a suitable Linux signal.

A subtle but critical invariant of the daemon concerns the order in which
the compiled programs are started.
eBPF programs attached to LSM hooks start intercepting system calls the
moment they are linked to their hook, so if the daemon attached a
\texttt{deny} or \texttt{kill} policy before configuring its set of target
processes, the policy would momentarily apply to every process on the host.
To avoid this race, the daemon enforces a three-phase initialization:
it first \emph{loads} each BPF object into the kernel, creating the maps, 
then \emph{populates the target map} with the PIDs or cgroup identifiers, and
only at the end \emph{attaches} the programs to their hooks.
The same invariant is preserved across policy reloads: new programs are
loaded, and their target map is transferred from the old instance before
the old programs are detached, so there is no time window in which an
unconfigured monitor is live.

The user can control how much space is allocated to the monitor memory, either directly or through one of four available profiles.
The daemon applies this choice before the loading phase and refuses to start when the requested capacity is below the requirements recorded by the compiler in the generated object.

The compiler and the daemon also implement two countermeasures against techniques an attacker may use to evade a policy.
These countermeasures are a first step towards strengthening the threat model of \Cref{sec:threat}, which assumes an attacker that does not actively evade the deployed policies.
The first targets evasion through \texttt{fork}: a policy scoped to \texttt{pid} correlates events from a single process, so an attacker who splits a kill chain across a \texttt{fork} would present the monitor with two unrelated processes.
The compiler, therefore, emits additional programs that propagate a process's history to its children and reclaim it when the process exits.
The second targets evasion through file aliasing: a path is one of the names of a file rather than the file itself, so symbolic and hard links, bind mounts, and renames give an attacker distinct names for the same file.
When a policy mentions literal paths, the compiler keys the match on the file's identity and keeps it up to date as files are created, aliased, or removed.

Beyond the core constructs of
\Cref{sec:lang}, the implementation supports further language features to simplify policy creation.
It provides four additional temporal operators:
\emph{bounded sequences} (\texttt{sequence $P_1$ then within $\Delta_1$
$P_2$ \dots}) for multi-step actions with per-step deadlines,
\emph{bounded until} (\texttt{while $P_1$ then within $\Delta$
$P_2$}) for asserting that a condition $P_2$ must occur within
$\Delta$ once $P_1$ holds, 
\emph{bounded count} (\texttt{count $E \geq N$ within $\Delta$}) for
detecting $N$ or more occurrences of an event in a sliding window, and \emph{bounded safety} (\texttt{forbid P when G within $\Delta$}).
Bounded sequences desugar into independent bounded-response clauses of the base temporal logic. 
Bounded until, count, and bounded safety are compiled via dedicated schemes.
To promote reuse, the DSL supports \emph{parameterized templates}
and \emph{policy inheritance}.
A policy can declare parameters that are bound at instantiation time via
the \texttt{use} construct, allowing a single template to be reused
across different policies.
A policy can also \emph{extend} another, inheriting its clauses and
optionally adding or specializing them.
Both mechanisms are resolved at desugaring time and do not affect the
formal semantics of the language.
In addition, \thetool ships a standard library of event schemas covering common
Linux subsystems (file operations, process lifecycle, network
connections, raw kernel events) and reusable policy templates, e.g., Chinese Wall, sandboxing, data-loss prevention, and rate limiting.
More details on all additional constructs, their compilation, and the standard library are in \extref{app:extensions}.

\subsection{Experimental Evaluation}
\label{sec:eval}

Our experimental evaluation addresses three research questions:
\begin{description}[nosep,leftmargin=1.5em,labelindent=0em]
  \item[\textbf{RQ1.}] What overhead does a \thetool monitor
    impose on hooked system calls, and how does it scale with policy
    complexity?
  \item[\textbf{RQ2.}] Can \thetool monitors block real
    attack patterns while avoiding false positives on benign workloads?
  \item[\textbf{RQ3.}] What policies can \thetool easily express and enforce that existing tools cannot?
\end{description}

All experiments ran on a laptop with an Intel Core
i7-1195G7 (4 cores, 8 threads, $2.90$\,GHz base clock), 32\,GB RAM,
Ubuntu 24.04, Linux kernel 6.8.0.
All runs were performed with the CPU frequency pinned to the base clock (governor
\texttt{performance}, turbo disabled).

\subsubsection{RQ1: Runtime Overhead}
\label{sec:eval:perf}

We evaluated \thetool on five micro-benchmarks (B-1 to B-5) and
one macro-benchmark (B-6).

\paragraph{Micro-benchmarks}
The micro-benchmarks use a synthetic workload consisting of a C
program that invokes \texttt{open}/\texttt{close} syscalls $N{=}100{,}000$ times.
The monitor execution time is measured via the kernel's built-in
profiling infrastructure, which maintains, for each loaded BPF
program, a cumulative invocation counter (\texttt{run\_cnt}) and a
cumulative execution-time counter (\texttt{run\_time\_ns}).
By reading both counters before and after the $N$ syscalls, we
compute the per-invocation BPF cost as
$\Delta\texttt{run\_time\_ns} / \Delta\texttt{run\_cnt}$,
where $\Delta$ denotes the difference between the two readings,
obtaining a direct kernel-level measurement without user-space
timing noise.
We repeat this process for $R{=}10$ trials and report mean~$\pm$~standard~deviation.
The overhead is computed as $\text{BPF cost} / \text{baseline cost} \times 100$,
where the baseline cost ($1318$\,ns/op) is measured via \texttt{clock\_gettime} without any BPF program loaded.
The baseline was measured again after each run: it drifted by less than
$1\%$, and the CPU never throttled.
We generate synthetic policies attached to the \texttt{security\_file\_open} LSM hook.
Each policy defines an event \texttt{open} with a string field (the file path), extracted via the kernel helper \texttt{bpf\_d\_path}.

\begin{table}[t]
\caption{Monitor execution time per
  \texttt{open}/\texttt{close} syscall (baseline without monitor:
  $1318$\,ns/op).}\label{tab:overhead}
\centering\small
\begin{tabular}{@{}llrr@{}}
\toprule
 & \textbf{Configuration} & \textbf{BPF (ns/op)} & \textbf{Overhead} \\
\midrule
B-1 & $C{=}1$ clause        & $240 \pm 2$ & $18.2\%$ \\
B-1 & $C{=}5$ clauses       & $249 \pm 2$ & $18.9\%$ \\
B-1 & $C{=}20$ clauses      & $295 \pm 2$ & $22.4\%$ \\
\midrule
B-2 & $H{=}1$ predicate     & $260 \pm 1$ & $19.7\%$ \\
B-2 & $H{=}5$ predicates    & $341 \pm 2$ & $25.9\%$ \\
B-2 & $H{=}10$ predicates   & $432 \pm 2$ & $32.8\%$ \\
\midrule
B-3 & $K{=}1{-}100$ targets   & $61 \pm 1$  & $4.6\%$ \\
B-4 & No path extraction      & $71 \pm 1$  & $5.4\%$ \\
B-5 & Non-targeted process    & $61 \pm 1$  & $4.6\%$ \\
\bottomrule
\end{tabular}
\end{table}

\Cref{tab:overhead} reports the results of these benchmarks.
The first two measure how the cost scales with the size of the policy.
In benchmark B-1, we vary the number of \texttt{forbid} clauses ($C$) and show
that the cost is \emph{nearly constant}: from $C{=}1$ to
$C{=}20$, the monitor execution time grows from $240$ to $295$\,ns
(${\sim}3$\,ns per additional clause).
Benchmark B-2 varies the number of historical predicates
($H$) in a policy and shows a \emph{linear} cost: each predicate requires one
BPF hash-map lookup, contributing ${\sim}19$\,ns (from $260$\,ns at $H{=}1$ to $432$\,ns at $H{=}10$).
The other three benchmarks allow us to understand where the fixed cost comes from. Benchmark B-4 repeats the measurement of B-1 at $C{=}1$ with an event taking an integer parameter, eliminating the call to \texttt{bpf\_d\_path} and the checks that
accompany it: the cost drops to $71$\,ns, showing that resolving the file path
accounts for ${\sim}169$\,ns, about $70\%$ of the per-invocation cost.
Benchmarks B-3 and B-5 measure a process that the policy does not target:
the generated code checks the target map and returns \emph{before} extracting
the event, so neither reaches \texttt{bpf\_d\_path}.
Benchmark B-3 varies the number of scoped targets ($K{=}1$ to $100$) on the same
path-based schema as B-1 at $C{=}1$. It confirms that the target lookup is constant in $K$: all configurations yield $61 \pm 1$\,ns.
Benchmark B-5 repeats the measurement on the integer schema of B-4, paying only $61$\,ns.
In all benchmarks except B-3 and B-5, we use global scope, which is a worst case because every
process pays the hook cost.

\paragraph{Macro-benchmark}
To assess the impact on an application workload, we deploy a
non-trivial policy on nginx serving a 1\,KB static file.
We generate sustained HTTP load using \texttt{wrk} with 2 threads
and 100 connections, running each trial for 30\,s and repeating
for $R{=}10$ trials.
The load generator and the server are pinned to disjoint physical cores to prevent the measurement from being dominated by contention between them.
The considered policy combines the post-exploitation sandbox
(CS-1) with SSH lateral-movement detection (CS-5), both described in the next subsection.
During normal HTTP traffic, nginx never spawns a shell or reads
SSH keys, so all historical predicates remain false throughout the
test; the monitor still executes on every hooked syscall, but
every guard evaluates to false and no violation is triggered.
This represents the common case in production, where the monitor
is active, but no attack is in progress.
Without any policy, nginx sustains
$155{,}798 \pm 666$\,req/s; with the policy loaded, throughput
drops to $147{,}223 \pm 436$\,req/s, a reduction of $5.5\%$.
Using the same kernel counters as above, the per-invocation BPF cost is
$149$\,ns, consistent with the micro-benchmark results for policies
of comparable complexity.

\medskip
Note that the per-invocation cost is not specific to \thetool, as every BPF-LSM monitor pays this cost whenever a hooked operation occurs. 
In fact, the overhead of \thetool is similar to that of other monitors that rely on LSM hooks.
For example, BPFContain~\cite{findlay2021bpfcontain} reports $12.69$--$14.31\%$ on Apache and up to $16.76\%$ on micro-benchmarks, and BPFBox~\cite{findlay2020bpfbox} $1.0$--$15.0\%$ across the Phoronix suite.
For the tools of \Cref{sec:eval:comparison}, an independent study~\cite{her2025ebpf} reports $9.79\%$ for Tetragon and $9.96\%$ for Falco on a per-execution micro-benchmark, while their vendors report lower overheads.
Our results, $5.5\%$ on nginx and $4.6$--$5.4\%$ per hooked syscall when the hook does not resolve a file path (B-3--B-5), sit at the low end of this range.
Our comparison is indicative only, as the overheads above were measured across different workloads and machines.

\subsubsection{RQ2: Attack Surface Reduction}
\label{sec:eval:cases}

We evaluated the effectiveness of \thetool on the seven case studies (CS-1 to CS-7) of \Cref{tab:casestudies}.
Each policy targets a post-exploitation syscall pattern associated with a documented attack or CVE (third column); the original exploit is outside the scope of the policy.
For each case study, we execute two trigger programs: an \emph{attack trigger} that
reproduces the syscall pattern, and a \emph{benign trigger} that performs similar but legitimate operations.

\begin{table}[t]
\caption{Case studies: attack surface reduction.
  \emph{Outcome}: how the attack trigger is handled.
  In all cases the benign trigger completes without interference.}\label{tab:casestudies}
\centering\small
\begin{tabular}{@{}clp{2.5cm}c@{}}
\toprule
 & \textbf{Scenario} & \textbf{Motivation} & \textbf{Outcome} \\
\midrule
CS-1 & Post-exploit.\ sandbox  & Log4Shell~\cite{cve202144228}         & Killed  \\
CS-2 & Resource exhaustion     & Fork-bomb~\cite{mitreT1499}           & Alerted \\
CS-3 & Container isol.\ (BLP) & CVE-2024-21626~\cite{cve202421626}    & Denied  \\
CS-4 & Data exfiltration       & Equifax 2017~\cite{equifax2018gao}    & Denied  \\
CS-5 & 3-step SSH kill chain   & MITRE T1021~\cite{mitreT1021}         & Killed  \\
CS-6 & Persistence + FD leak    & Persistence~\cite{mitreT1053}     & Killed  \\
CS-7 & Clark-Wilson (E2)       & Dirty Pipe~\cite{cve20220847}         & Denied  \\
\bottomrule
\end{tabular}
\end{table}

In all cases, the attack trigger is correctly blocked or
alerted, whereas the benign trigger completes without interference (zero false negatives and zero false positives).
In CS-1 (\emph{post-exploitation sandbox}), the policy detects when a server process
spawns a shell and then prevents that shell from executing further
binaries, opening network connections, or writing configuration files. 
This behavioral pattern is common in post-exploitation steps.
The policy relies on a historical predicate to condition the enforcement on a previous \texttt{exec}.
In CS-2 (\emph{resource exhaustion}), the policy raises an alert when a process clones itself more than $100$ times within a $1\,\text{s}$ time window, mitigating fork-bomb attacks, or opens more than $50$ new connections within a $10\,\text{s}$ sliding time window, mitigating connection-flooding patterns.
CS-3 (\emph{container isolation}) enforces a Bell--LaPadula no-read-up /
no-write-down policy across container cgroups, preventing
information flows from a high-sensitivity container to a low one via the host file system.
In CS-4 (\emph{data exfiltration}), the policy detects when a process reads a
sensitive file and \emph{subsequently} opens a network connection; it correlates two events on different LSM hooks via cross-hook binding.
Case study CS-5 (\emph{SSH lateral movement}) was already presented in \Cref{sec:overview}: the policy detects a three-step kill chain in which the attacker reads an SSH private key, opens an outbound connection on port~22, and finally executes a command.  The process is killed on the final \texttt{exec}.
In CS-6 (\emph{persistence containment + FD-leak watchdog}), the
policy combines two clauses with different enforcement
semantics: an inline safety check kills the process if it attempts
post-exploitation persistence operations, e.g.,
\texttt{exec(/usr/bin/crontab)}, or write to \texttt{/etc/passwd},
after a shell has been spawned, while a liveness
obligation alerts when a file descriptor is opened but not closed
within $30\,\text{s}$, mitigating resource-leak patterns.
Case study CS-7 (\emph{Clark--Wilson integrity}) enforces the E2 rule of the Clark--Wilson model: protected files (e.g., 
\texttt{/etc/passwd}, setuid binaries) can be modified only by authorized processes. 
Any write from an unauthorized process is denied, blocking attacks such as Dirty Pipe, in which an unprivileged process bypasses standard permissions to overwrite a protected file.  The policy uses a negated history guard of the form ``forbid write unless an authorized program was executed''.

Additionally, the artifact includes six further case studies (CS-8 to CS-13), in which the attack trigger is a public proof-of-concept exploit for a documented (patched upstream) kernel vulnerability, and the policy forbids the system call patterns that may lead to the attack, e.g., the use of exotic crypto sockets.
Since a policy describes the attack pattern rather than a specific bug, it may intercept multiple vulnerabilities. For example, the policy of CS-8 covers \emph{copy\_fail}~(CVE-2026-31431) and \emph{Fragnesia}~(CVE-2026-46300), which affect different kernel subsystems but share the same behavioral pattern: the setup of an exotic crypto socket followed by the execution of a setuid binary.
These case studies show that \thetool may also serve as a containment layer for this class of bugs when the patch is not yet deployed.

\subsubsection{RQ3: Comparison with Existing Tools}
\label{sec:eval:comparison}

To compare \thetool with production tools, we took four real rules
from Falco~\cite{falco2023} and Tetragon~\cite{tetragon2024}, two
well-known open-source projects, and re-expressed each of them as a
\thetool policy.
The two Falco rules are \emph{Read sensitive file
untrusted}~\cite{falcoRuleReadSensitive} and \emph{Write below
binary directories}~\cite{falcoRuleWriteBin}: both belong to the
default ruleset shipped with every Falco installation and are among
the most cited examples in the project's documentation and tutorials,
targeting credential exfiltration and binary tampering, respectively.
The two Tetragon policies are
\texttt{monitor-kernel-modules}~\cite{tetragonKmodPolicy} and
\texttt{sys\_mount}~\cite{tetragonMountPolicy}, both taken from the
examples in the upstream repository, and cover two attack surfaces that the documentation prominently
advertises as use cases: kernel-rootkit installation and mount-based
container escape.
For each rule, we provide a \emph{direct port policy} that reproduces the
upstream behavior, and one or more \emph{companion policies} that
capture behavior the upstream rule cannot express.
The policies and their attack and benign triggers are in the online
repository.

Falco's \emph{Read sensitive file untrusted} ($\sim 35$ LoC) 
matches reads of credential, SSH-key, TLS-key, and Kubernetes-secret
files performed outside an allowlist of trusted readers, and raises
an alert.
Our direct port policy ($\sim 16$ LoC) keeps the same set of paths, but instead of a tracepoint uses the LSM hook \texttt{security\_file\_open}, so
the same event can be synchronously denied.
A companion policy denies the same reads only \emph{after} the
offending process has spawned a shell, capturing a ``shell exec~$\to$~credential read'' post-exploitation kill
chain~\cite{mitreT1552}. 
Falco cannot express this stateful correlation, since its rule language is single-event and stateless.
For \emph{Write below binary directories} ($\sim 19$ LoC), 
we follow the same approach: a direct port ($\sim 16$ LoC) denies writes to
\texttt{/bin}, \texttt{/sbin}, \texttt{/usr/bin}, \texttt{/usr/lib},
and similar locations, while the companion policy kills the offending process when it attempts these writes after a shell spawn.

Tetragon's \texttt{monitor-kernel-modules} policy ($\sim 50$ LoC) relies on the 
\texttt{security\_kernel\_module\_request} and \texttt{do\_init\_module} hooks:
it logs every attempt to load a kernel module but does not block it.
Our direct port policy ($\sim 7$ LoC) relies on the same hooks but upgrades the action
to a synchronous LSM deny, so the module is rejected before it is
linked into the kernel.
We provide two companion policies that perform cross-hook correlations that a single Tetragon policy cannot express.
The first synchronously blocks the \texttt{bpf()} syscall
when invoked with the argument \texttt{BPF\_PROG\_LOAD} for loading a BPF program or with the argument \texttt{BPF\_MAP\_CREATE} for creating a BPF map;
the second policy blocks any \texttt{bpf()} syscall \emph{after} the same
process has loaded a kernel module, capturing in a single policy a kill chain in which an attacker first installs a kernel-mode rootkit~\cite{mitreT1014} and then leverages BPF for further in-kernel persistence.
Tetragon's \texttt{sys\_mount} policy ($\sim 12$ LoC) traces every
\texttt{mount} syscall and reports its source and target without
filtering or blocking.
Our direct port policy ($\sim 11$ LoC) attaches to the LSM hook \texttt{security\_sb\_mount}
and synchronously denies mounts of the filesystem types most commonly
abused for container escape, such as \texttt{proc} and \texttt{sysfs}.
The two companion policies extend it, considering multiple hooks: the first denies any
\texttt{bpf()} syscall after a suspicious mount has been observed,
capturing a kill chain in which an attacker first escapes the
container by mounting host filesystems~\cite{mitreT1611} and then uses BPF to consolidate the
breakout; the second denies every outbound \texttt{connect()} after
the process has spawned a shell, capturing post-escape lateral
movement~\cite{mitreT1021}.

We tested all four ports by providing an attack and a benign trigger
for each policy, following the same approach as in RQ2: in every case
the attack trigger is blocked, while the benign trigger completes
without interference.
The following observations emerge from the comparison.
\thetool policies are typically shorter. 
Falco rules are intrinsically alert-only and operate on a single hook,
because they observe syscalls via tracepoints \emph{after} the kernel
has executed them; re-expressing them in \thetool turns detection
into synchronous prevention with no change to the matching logic,
simply by attaching the same predicate to the corresponding LSM hook.
Tetragon supports synchronous enforcement, but its policies are
single-hook by construction, so cross-event correlations such as
``module load \emph{then} BPF program load'' or ``shell exec
\emph{then} connect'' cannot be expressed in a single policy,
whereas \thetool's historical predicates and \texttt{when} construct express
these correlations directly inside one policy.
The comparison also revealed two matching primitives that Tetragon
offers and \thetool does not: testing whether a specific bit is set
in a syscall flag argument, and filtering on the privileges of the
calling process. 
We leave their implementation in \thetool for future work.

%% file: sections/related-work.tex

We discuss the two lines of research to which our work is closely related: in-kernel security monitoring and runtime verification.

\paragraph{In-kernel security monitoring}
Falco~\cite{falco2023}, Tetragon~\cite{tetragon2024}, and Kube\-Armor~\cite{kubearmor} are the industry references for eBPF-backed runtime security: they express policies as YAML
rules whose semantics is defined only by the implementation, and they
do not statically distinguish the hooks on which blocking is possible from those on which it is not. 
We compared \thetool against Falco and Tetragon in \Cref{sec:impl} (RQ3). 
The tool bpftrace~\cite{bpftrace} is complementary: a scripting
language for kernel tracing focused on observability, with no
LSM hook support and with write access restricted to error-injection points.
BPFBox~\cite{findlay2020bpfbox} and BPFContain~\cite{findlay2021bpfcontain} compile process
confinement policies to eBPF LSM programs, demonstrating the feasibility of the hook point but without temporal or stateful operators. 
Jia et al.~\cite{jia2023} extend Seccomp-BPF with temporal specialization and counter-based checks at the system-call boundary; both are subsumed by \thetool's policy language.
More recently, eBPF-PATROL~\cite{patrol2025} applies eBPF to
runtime threat recognition in containerized environments, again
without formal guarantees or a type discipline.
The closest system in spirit is the Linux kernel's own runtime
verification subsystem~\cite{bristot2019linuxrv,linuxrv-docs},
upstream since v6.0, which loads hand-drawn deterministic
automata as kernel modules attached to tracepoints and reacts to violations
with \texttt{printk} or \texttt{panic}. 
\thetool shares the same goal of being an online monitor inside the kernel, but
offers a DSL for the specification of policies with a formal
semantics and a type system that statically separates
controllable from merely observable hooks, delivers monitors as
eBPF programs rather than kernel modules, attaches to LSM hooks as well, and scopes enforcement per process, namespace, or cgroup.

Overall, none of these systems combines a policy language with formal semantics, a type discipline separating controllable from observable events, and provably correct compilation to in-kernel monitors,  which is what \thetool contributes.

\paragraph{Runtime verification}
\thetool brings runtime verification inside the kernel. 
The construction of finite-state monitors from
temporal-logic specifications dates back to the work of
Bauer et al.~\cite{bauer2011} for LTL and
TLTL, and to the automata-based techniques of
Havelund and Ro\c{s}u~\cite{havelund2004}.
On the theoretical side, Schneider~\cite{schneider2000} and Basin et al.~\cite{basin13} characterized the classes of policies enforceable by execution monitors;
Aceto et al.~\cite{aceto2023foe} extended enforcement to
first-order branching-time properties, and Hublet et
al.~\cite{hublet2024proactive,hublet2025scaling} developed
proactive enforcement for metric first-order temporal logic
together with scalable implementations.
At the tool level, MonPoly~\cite{basin2011} monitors metric
first-order temporal properties over logs,
JavaMOP~\cite{meredith2012javamop} compiles parametric
specifications into aspect-oriented instrumentation, and
DejaVu~\cite{havelund2018dejavu} targets first-order past-time
logic; decentralized variants such as RIARC~\cite{aceto2024riarc}
address the instrumentation of reactive components. 
All of these monitors rely on very expressive logics and run in user space, in the same address space as the monitored program, or offline over recorded traces; \thetool aims at in-kernel policy enforcement and realizes monitor synthesis under the constraints of the
eBPF verifier: bounded memory, no unbounded loops, a fixed
instruction budget.
In-kernel enforcement requires trading off the expressiveness of the policy language with implementation constraints, as done in \Cref{sec:compilation}.

%% file: sections/conclusion.tex

We have presented \thetool, a runtime-verification
framework to enforce trace-based security policies inside the Linux kernel.
\thetool offers a policy language with a formal semantics over execution traces and a type system
that statically separates controllable from observable events.
It then synthesizes, from every well-typed policy, a correct finite-state monitor, and compiles it to eBPF programs that run synchronously inside the kernel. 
We implemented a PoC that includes a compiler and a runtime daemon to manage the policy life-cycle.
We have evaluated our implementation on seven case studies drawn from real-world attack patterns and on micro- and macro-benchmarks to show that stateful security policies  can be enforced inside the kernel with an acceptable overhead.
Finally, we have qualitatively evaluated the expressiveness of our policy language against upstream Falco and Tetragon rules, showing that our language can prevent more complex attack patterns than these widely adopted tools.

\thetool inherits two structural limitations from its target environment.  
First, synchronous enforcement is only possible on events mediated by an LSM hook: policies whose violating events occur only through tracepoints or kprobes can be
\emph{detected} but not \emph{prevented}.  
Second, the eBPF execution model forbids dynamic allocation; as a consequence, the monitor state is stored in fixed-size maps, whose capacity is set when a policy is loaded and checked against the requirements computed at compile time.  

There are several directions in which to extend the present work.  
On the \emph{language} side, we plan to introduce richer predicates and inline aggregations, and lift correlation from kernel identities (\texttt{pid}, \texttt{tgid}, \texttt{cgroup}) to causal ones (process trees and data-flow taint): the propagation of history across \texttt{fork} (\Cref{sec:impl:tool}) is a first step, but attacks that spread through \texttt{execve} or across cooperating processes require the unit of monitoring itself to be defined by causality rather than by a declared scope.
On the \emph{formal} side, we are interested in mechanizing our soundness proof and in monitoring hyperproperties such as non-interference, which lie beyond plain trace properties.  
On the \emph{deployment} side, we plan to explore distributed enforcement across a cluster of containers, automatic policy synthesis from high-level intents or from observed benign behavior, and integration with container orchestrators to tie the life-cycle of policies to that of the workloads they protect.

%% file: sections/appendix-formal.tex
\subsection{Policy Language}

\subsubsection{Well-formedness Conditions}\label{sec:well:form}
We define well-formedness conditions for the temporal logic of \Cref{sec:logic}.
We first define the function $\bnd$ to compute the set of variables whose bindings propagate to the enclosing context in a guard $g$, an event predicate $P$, and a temporal formula $\phi$ (below, $\alpha.g$ denotes the guard of $\alpha$):  
\begin{gather*}
\bnd(\mathsf{true}) = \bnd(\field \sim \ell) = \emptyset \quad \bnd(\field = {?X}) = \{X\} \\
\bnd(\field = X) = \emptyset \quad \bnd(g_1 \land g_2) = \bnd(g_1) \cup \bnd(g_2)
\\
\bnd(\alpha) = \bnd(\alpha.g) \quad \bnd(\lnot P) = \bnd(P)
\\
\bnd(P_1 \land P_2) = \bnd(P_1 \lor P_2) = \bnd(P_1) \cup \bnd(P_2)
\\
\bnd(\phi_1 \land \phi_2) = \bnd(\phi_1) \cup \bnd(\phi_2) \quad \bnd(\neg\phi) =\bnd(\phi)
\\
 \bnd({<}\phi) = \bnd(G\,\phi) = \bnd(P_1\,\Fbound\,P_2) = \emptyset
\end{gather*}
Note that for the operators $G\,\phi$, ${<}\phi$, and $\Fbound$, the set of variables is
$\emptyset$ due to an intrinsic property of each operator: in all three cases,
the semantics returns the environment $\rho$ unchanged --- $G$ and $\Fbound$ return $\rho$ invariant explicitly in their semantic clauses, and ${<}\phi$ evaluates $\phi$
under the empty environment and returns $\rho$ unchanged. 
No new binding reaches the outer context, so computing $\emptyset$ faithfully
reflects the semantics.

\begin{figure*}[t]
\begin{gather*}
\inference[\rulename{GTrue}]{ }{ \rho \vdash \mathsf{true}\ \mathsf{wf} } 
\qquad 
\inference[\rulename{GLit}]{ }{\rho \vdash \field \sim \ell\ \mathsf{wf}} 
\qquad
\inference[\rulename{GBind}]{ X \notin \mathsf{dom}(\rho) }{ \rho \vdash \field = {?X}\ \mathsf{wf}} 
\qquad
\inference[\rulename{GConstr}]{ X \in \mathsf{dom}(\rho) }{\rho \vdash \field = X\ \mathsf{wf}} 
\\[1ex]
\inference[\rulename{GAnd}]{\rho \vdash g_1\ \mathsf{wf} \quad \rho \cup \bnd(g_1) \vdash g_2\ \mathsf{wf}}{\rho \vdash g_1 \land g_2\ \mathsf{wf}} 
\qquad\quad
\inference[\rulename{PAtom}]{ \rho \vdash \alpha.g\ \mathsf{wf} }{\rho \vdash \alpha\ \mathsf{wf}} 
\qquad\quad
\inference[\rulename{PNeg}]{\rho \vdash P\ \mathsf{wf} \quad \bnd(P) = \emptyset}{\rho \vdash \lnot P\ \mathsf{wf}} 
\\[1ex] 
\inference[\rulename{PAnd}]{\rho \vdash P_1\ \mathsf{wf} \quad \rho \cup \bnd(P_1) \vdash P_2\ \mathsf{wf}}{\rho \vdash P_1 \land P_2\ \mathsf{wf}}
\qquad
\inference[\rulename{Por}]{\rho \vdash P_1\ \mathsf{wf} \quad \rho \vdash P_2\ \mathsf{wf} \\
	\bnd(P_1) = \emptyset \quad \bnd(P_2) = \emptyset}{\rho \vdash P_1 \lor P_2\ \mathsf{wf}} 
\qquad
\inference[\rulename{FAnd}]{\rho \vdash \phi_1\ \mathsf{wf} \quad \rho \cup \bnd(\phi_1) \vdash \phi_2\ \mathsf{wf}}{\rho \vdash \phi_1 \land \phi_2\ \mathsf{wf}}
\\[1ex]
\inference[\rulename{FNeg}]{\rho \vdash \phi\ \mathsf{wf} \quad \bnd(\phi) = \emptyset}{\rho \vdash \neg\phi\ \mathsf{wf}}
\qquad
\inference[\rulename{FPast}]{\bnd(\phi) \cup \use(\phi) = \emptyset}{\rho \vdash {<}\phi\ \mathsf{wf}} 
\qquad
\inference[\rulename{FAlways}]{\rho \vdash \phi\ \mathsf{wf}}{\rho \vdash G\,\phi\ \mathsf{wf}} 
\qquad
\inference[\rulename{FResp}]{\rho \vdash P_1\ \mathsf{wf} \quad \rho \cup \bnd(P_1) \vdash P_2\ \mathsf{wf}}{\rho \vdash P_1\,\Fbound\,P_2\ \mathsf{wf}}
\end{gather*}
\caption{Well-formedness conditions for guards, event predicates, and temporal formulas.}
\label{fig:well:form}	
\end{figure*}

\Cref{fig:well:form} shows the well-formedness conditions for guards, event predicates, and temporal formulas (we denote by $\rho \cup V$ the environment obtained by extending $\rho$ with a binding to a placeholder value for each variable in $V$).
A guard $g$ is well formed under an environment $\rho$, written $\rho \vdash g\ \mathsf{wf}$, iff it satisfies the rules \rulename{G*}.
Intuitively, for each sub-term $g_1 \land g_2$ where $g_2$ contains a constraint $\field = X$, it must hold that $X \in \mathsf{dom}(\rho) \cup \bnd(g_1)$.
Similarly, an event predicate $P$ is well formed under an environment $\rho$, written $\rho \vdash P\ \mathsf{wf}$, iff it satisfies the rules \rulename{P*}.
Intuitively, these rules enforce the following conditions:
for each conjunction $P_1 \land P_2$ where $P_2$ includes a constraint $\mathsf{field} = X$, it must hold that $X \in \mathsf{dom}(\rho) \cup \bnd(P_1)$;
no binding $?X$ may be introduced under disjunction $\lor$.

A temporal formula $\phi$ is well formed under an environment $\rho$, in symbols $\rho \vdash \phi\ \mathsf{wf}$, iff it satisfies the rules \rulename{F*}. 
Intuitively, rule \rulename{FAnd} requires that for each conjunction $\phi_1 \land \phi_2$ where $\phi_2$ includes a constraint $\mathsf{field} = X$, it must hold that
$X \in \mathsf{dom}(\rho) \cup \bnd(\phi_1)$: bindings introduced
by $\phi_1$ are visible in $\phi_2$.
Rule \rulename{FNeg} has two independent premises: $\rho \vdash \phi\ \mathsf{wf}$ ensures that $\phi$ is well formed under $\rho$, and $\bnd(\phi) = \emptyset$ ensures that no binding escapes from the negated formula. 
These premises are independent: a formula can be well formed and still produce bindings.
The second premise is necessary because under negation a binding may never actually be
produced: if $\phi$ evaluates to $\violated$, no field was matched and no variable was bound. 
Allowing bindings to escape from $\lnot\phi$ would therefore be unsound.
Rule \rulename{FPast} requires only $\bnd(\phi) \cup \use(\phi) = \emptyset$, i.e., $\phi$ neither introduces nor uses correlation variables. 
A formula with no variables is well formed under any  environment, so the explicit premise $\rho \vdash \phi\ \mathsf{wf}$ would be redundant and is omitted. 
The condition ensures that the semantics of ${<}\phi$ evaluates $\phi$ on past prefixes under the empty environment, independently of the current $\rho$. 
If $\phi$ referred to variables bound in $\rho$, those bindings would be
unavailable at past positions, making the evaluation ill-defined.
Rule \rulename{FAlways} is direct: we check $\phi$ is well formed.
Rule \rulename{FResp} mirrors \rulename{FAnd}: the bindings
introduced by $P_1$ are visible in $P_2$, expressed by the premise
$\rho \cup \bnd(P_1) \vdash P_2\ \mathsf{wf}$.

\subsubsection{Event Predicate Semantics}\label{sec:pred:sem}
Let $e$ be an event and $\rho \colon \Vars \to \Vals$ a correlation environment.
The semantics of an event predicate is defined by the function 
$ 
\pred{\cdot} \colon \mathcal{P} \times \mathcal{E} \times \mathsf{Env} \;\to\; K_3 \times \mathsf{Env} 
$,
where $\mathcal{P}$ is the set of event predicates, $\mathcal{E}$ is the set of events, $\mathsf{Env}$ is the set of correlation environments, and $K_3$ is the set $\{\mtrue, \mfalse, \mnoapp \}$. 
The function $\pred{\cdot}$ is defined inductively on the syntax of event predicates as follows (below $\sim$ is an operator in $\{=, \neq, <, \leq, >, \geq\}$):

\smallskip
\noindent\emph{Atoms:} Given an event atom $\alpha = (E, g)$:
	\[
	\pred{\alpha}(e,\, \rho)
	\;=\;
	\begin{cases}
		\llbracket g \rrbracket(e,\, \rho) & \text{if } e.\mathit{type} = E \\
		(\mnoapp,\, \rho)           & \text{otherwise}
	\end{cases}
	\]
	where $\llbracket g \rrbracket(e, \rho) = (b, \rho')$ is inductively defined as:
	\begin{align*}
		\llbracket \mathsf{true} \rrbracket(e, \rho)&= \bigl(\mtrue,\;\rho\bigr)\\
		\llbracket \mathsf{field} = {?X} \rrbracket(e, \rho)
		&= \bigl(\mtrue,\;
		\rho[X \mapsto e.\mathsf{params}[\mathsf{field}]]\bigr) \quad \text{if } X \notin \mathsf{dom}(\rho) \\
		\llbracket \mathsf{field} = X \rrbracket(e, \rho)
		&= \bigl(e.\mathsf{params}[\mathsf{field}] = \rho(X),\; \rho\bigr) \quad \text{if } X \in \mathsf{dom}(\rho) \\
		\llbracket \mathsf{field} \sim \ell \rrbracket(e, \rho)
		&= \bigl(e.\mathsf{params}[\mathsf{field}] \sim \ell,\; \rho\bigr) 
		 \\ 
		\llbracket g_1 \land g_2 \rrbracket(e, \rho)
		&= (b_1 \land b_2,\; \rho_2), 
		\quad (b_i, \rho_i) = \llbracket g_i \rrbracket(e,\rho_{i-1}),\;
		\rho_0 = \rho
	\end{align*}

\smallskip
\noindent\emph{Negation:} 
	\[
	\pred{\lnot P}(e,\, \rho)
	\;=\; 
	\begin{cases}
		(\lnot b, \rho) & \text{if } b = \pi_1\left(\pred{P}(e,\rho)\right) \in \{\mtrue, \mfalse\}\\
		(\mnoapp, \rho) & \text{otherwise}
	\end{cases}
	\]
	Note that we do not add new bindings to the environment $\rho$: the well-formedness condition ensures that $\lnot P$ introduces no new bindings.

\smallskip	
\noindent\emph{Conjunction:}
\[
\pred{P_1 \land P_2}(e,\, \rho) \;=\; r 
\]
where $(b_1, \rho_1) = \pred{P_1}(e, \rho)$ and 
\[
 r \;=\;
 \begin{cases}
 	(\mfalse, \rho) & \text{if } b_1 = \mfalse \\
 	(\mtrue, \rho_2)  & \text{if } b_1 = \mtrue \land (\mtrue, \rho_2) = \pred{P_2}(e, \rho_1) \\
 	(\mfalse, \rho) & \text{if } (\mfalse, \rho_1) = \pred{P_2}(e, \rho_1)\\
 	(\mnoapp, \rho) & \text{otherwise} 
 \end{cases}
\]
Note that the bindings introduced by $P_1$ are visible in $P_2$ (the conjunction is not commutative). 

\smallskip
\noindent\emph{Disjunction:}
\[
\pred{P_1 \lor P_2}(e,\, \rho) \;=\; (v,\; \rho)
\]
where $(b_1, \rho) = \pred{P_1}(e, \rho)$, $(b_2, \rho) = \pred{P_2}(e, \rho)$, and $v$ is 
\[
v \;=\;
\begin{cases}
	\mtrue  & \text{if } b_1 = \mtrue \lor b_2 = \mtrue\\
	\mfalse & \text{if } b_1 = \mfalse \land b_2 = \mfalse\\
	\mnoapp & \text{otherwise}
\end{cases}
\]
Note that $P_1$ and $P_2$ contain no binding by the well-formedness condition, so the environment always remains the same. 
\smallskip

Hereafter, we write $\pred{P}(e,\rho) = v$, for $v \in K_3$, to mean
$\pi_1(\pred{P}(e,\rho)) = v$, when the resulting environment is immaterial.

\subsubsection{Temporal Formulas Semantics}\label{sec:app:log:sem}
The evaluation function $\sem{\cdot}$ is inductively defined over the syntax of temporal formulas.
Below, we report the semantics for all the constructs except for the bounded response, whose semantics is in \Cref{sec:lang}.
 
\noindent For an event predicate $P$:
\[
\sem{P}(\sigma_k, \rho) =
\begin{cases}
	(\ok,      \rho') &
	\text{if } \pred{P}(e_k, \rho) = (\mtrue, \rho') \\
	(\violated, \rho) & \text{if } \pred{P}(e_k, \rho) = (\mfalse, \rho) \\
	(\pending,  \rho) & \text{if } \pred{P}(e_k, \rho) = (\mnoapp, \rho)
\end{cases}
\]
In case of violation, $\rho$ does not change: no new binding occurs.
If $e_k$ is irrelevant to the current predicate $P$, the result is $\pending$: we postpone the decision.

\noindent For the $\neg$ operator:
\[
\sem{\lnot \phi}(\sigma_k, \rho) 
\;=\;
\begin{cases}
	(\violated, \rho) & \text{if } \sem{\phi}(\sigma_k, \rho) = (\ok, \rho) \\
	(\ok, \rho) & \text{if } \sem{\phi}(\sigma_k, \rho) = (\violated, \rho) \\
	(\pending, \rho) & \text{otherwise}
\end{cases}
\]
By our well-formedness conditions, no bindings are introduced, so $\rho$ remains the same.

\noindent For conjunction:
\[
\sem{\phi_1 \land \phi_2}(\sigma_k,\, \rho)
\;=\; (v,\, \rho_2)
\]
where $(v_1, \rho_1) = \sem{\phi_1}(\sigma_k, \rho)$,
$(v_2, \rho_2) = \sem{\phi_2}(\sigma_k, \rho_1)$, and 
\[
v \;=\;
\begin{cases}
	\violated
	& \text{if } v_1 = \violated
	\text{ or } v_2 = \violated\\
	\pending
	& \text{if } v_1 = \pending
	\text{ or } v_2 = \pending \\
	\ok
	& \text{otherwise}
\end{cases}
\]
The bindings introduced by $\phi_1$ are visible in $\phi_2$.
We adopt an evaluation order from left to right to take care of bindings.

\noindent For the \emph{once} operator:
\[
\sem{{<}\phi}(\sigma_k, \rho)
\;=\; 
\begin{cases}
	(\ok, \rho) & \text{if } \exists i \leq k.\; \pi_1(\sem{\phi}(\sigma_i, \rho)) = \ok \\
	(\pending, \rho) & \text{otherwise}
\end{cases} 
\]
The once operator ${<}\phi$ holds on a prefix $\sigma_k$ if $\phi$ held
at some position $i \leq k$, including the current one. 
It never produces $\violated$: the absence of a past satisfaction yields $\pending$, not a violation, since no finite prefix can witness that $\phi$ will never hold. 
Consequently, ${<}\phi$ is monotone: once $\ok$, it remains $\ok$ on all subsequent prefixes. 
Note that the environment $\rho$ is returned unchanged: the
well-formedness condition ensures that $\phi$ neither introduces nor uses correlation
variables, so its evaluation is independent of $\rho$.

\noindent For the \emph{always} operator:
\[
\sem{G\,\phi}(\sigma_k,\rho) =
\begin{cases}
	(\violated,\rho) & \text{if } \exists i\leq k.\;
	\pi_1(\sem{\phi}(\sigma_i,\rho))=\violated \\
	(\ok,\rho)       & \text{otherwise}
\end{cases}
\]
where $\pi_1$ is the standard projection operator. This operator expresses safety properties: every violation is due to a finite bad prefix.
Consequently, $G\,\phi$ can only be $\ok$ or $\violated$, never $\pending$ because it does not require that something happen in the future.

\subsubsection{Formal Results}

\persi*
\begin{proof}
	We proceed by induction on $j - k \geq 0$.
	
	\medskip
	\noindent\emph{Base case ($j - k = 0$):} trivial, $\sigma_j = \sigma_k$ and the conclusion holds by hypothesis.
	
	\medskip
	\noindent\emph{Inductive step ($n = j - k > 0$):} As the inductive hypothesis,   assume the proposition holds for all pairs $(\sigma_k, \sigma_{j'})$ with
	$j' - k < n$.
	We need to show the proposition for $(\sigma_k, \sigma_j)$ with $j - k = n$.
	Since $n > 0$, the prefix $\sigma_{j-1}$ exists and $\sigma_k$ is a prefix of $\sigma_{j-1}$ with $(j-1) - k =
	n - 1 < n$.	
	By the inductive hypothesis applied to $(\sigma_k, \sigma_{j-1})$ we have:
	\[
	\sem{\pol}(\sigma_{j-1}) = \violated.
	\]
	Now, let $\phi$ be the denotation of the program $\pol$. By the definition of the semantic function, 
 $\phi$ has the form $\bigwedge_{\ell=1}^r G(\psi_\ell)$, which can be rewritten as $G(\psi)$ with $\psi = \bigwedge_{\ell=1}^r \psi_\ell$.
	By the semantics of $G$:
	\[
	\sem{G (\psi)}(\sigma_{j-1}, \emptyset) = \violated
	\implies
	\exists\, h \leq j-1.\;
	\pi_1(\sem{\psi}(\sigma_h, \emptyset)) = \violated.
	\]
	Since $h \leq j - 1 < j$, the same $h$ witnesses $\sem{G\,\psi}(\sigma_j, \emptyset) = \violated$, thus proving the thesis.
\end{proof}

\vioprefix*
\begin{proof}
	The set $S = \{i \leq k \mid \sem{\pol}(\sigma_i) = \violated\}$ is non-empty since $k \in S$ by hypothesis, and finite since $S \subseteq \{1, \ldots, k\}$. Hence $i^* = \min S$ exists and is unique.
\end{proof}
The following proposition ensures that a formula introducing no binding leaves the correlation environment unchanged:
\begin{proposition}[Environment invariance]
	\label{prop:env-inv}
	For every temporal formula $\phi$ such that $\bnd(\phi) = \emptyset$, every finite trace
	$\sigma_k$, and every environment $\rho$, there exists $v \in \B$ such that
	\[
	\sem{\phi}(\sigma_k, \rho) = (v, \rho).
	\]
\end{proposition}
\begin{proof}
	We proceed by structural induction on $\phi$.
	For the negation, once, always, and bounded-response cases, the thesis follows immediately
	from the corresponding semantic clauses, which return $\rho$ unchanged.
	For an event predicate $P$, $\bnd(P) = \emptyset$ means that $P$ contains no binding
	$\field = {?X}$, hence $\pred{P}(e_k, \rho)$ returns $\rho$ unchanged; by the semantics of
	predicates as formulas, $\sem{P}(\sigma_k, \rho)$ returns $\rho$ as well.
	For $\phi_1 \land \phi_2$ we have $\bnd(\phi_1) = \bnd(\phi_2) = \emptyset$; 
	by the induction hypothesis on $\phi_1$ we have that $\sem{\phi_1}(\sigma_k, \rho)$ leaves the environment unchanged, so $\phi_2$ is evaluated under $\rho$. 
	The thesis thus follows by the induction hypothesis on $\phi_2$.
\end{proof}
Below, we prove distributivity of $G$ over conjunction 
against the three-valued semantics of \Cref{sec:logic}.
\begin{proposition}[Distributivity of $G$ over $\land$]
	\label{prop:g-conj}
	For all temporal formulas $\phi_1, \phi_2$ such that $\bnd(\phi_1) = \emptyset$, all finite
	traces $\sigma_k$, and all environments $\rho$:
	\[
	\pi_1(\sem{G(\phi_1) \land G(\phi_2)}(\sigma_k, \rho))
	\;=\;
	\pi_1(\sem{G(\phi_1 \land \phi_2)}(\sigma_k, \rho)).
	\]
\end{proposition}
\begin{proof}
	Since $G$ always results in $\ok$ or $\violated$, never $\pending$, we show that  	the two sides agree on $\violated$; the $\ok$ case is similar. 
	By the semantics of $G$ and conjunction:
	\begin{align*}
		&\pi_1(\sem{G(\phi_1) \land G(\phi_2)}(\sigma_k, \rho))
		= \violated \\
		&\iff
		\pi_1(\sem{G(\phi_1)}(\sigma_k, \rho)) = \violated
		\;\lor\; \\
		& \qquad \quad \pi_1(\sem{G(\phi_2)}(\sigma_k, \rho)) = \violated \\
		&\iff
		\bigl(\exists\, i \leq k.\;
		\pi_1(\sem{\phi_1}(\sigma_i, \rho)) = \violated\bigr)
		\;\lor\; \\
		& \qquad \quad
		\bigl(\exists\, j \leq k.\;
		\pi_1(\sem{\phi_2}(\sigma_j, \rho)) = \violated\bigr)
	\end{align*}
	and, by \Cref{prop:env-inv}, the hypothesis $\bnd(\phi_1) = \emptyset$ ensures that in
	$\sem{\phi_1 \land \phi_2}(\sigma_i, \rho)$ the second conjunct is evaluated under $\rho$ as
	well, hence:
	\begin{align*}
		&\pi_1(\sem{G(\phi_1 \land \phi_2)}(\sigma_k, \rho))
		= \violated \\
		&\iff
		\exists\, i \leq k.\;
		\pi_1(\sem{\phi_1 \land \phi_2}(\sigma_i, \rho))
		= \violated \\
		&\iff
		\exists\, i \leq k.\;
		\pi_1(\sem{\phi_1}(\sigma_i, \rho)) = \violated
		\;\lor\; \\
		& \qquad\qquad \pi_1(\sem{\phi_2}(\sigma_i, \rho)) = \violated.
	\end{align*}
	The two conditions are logically equivalent by the
	propositional equivalence
	$
	\bigl(\exists\, i.\; A(i)\bigr) \lor
	\bigl(\exists\, j.\; B(j)\bigr)
	\;\iff\;
	\exists\, i.\; A(i) \lor B(i).
	$
\end{proof}

\begin{figure*}[t]
	\begin{gather*}
		\inference[\rulename{PEvent}]{ }{\Gamma \vdash (E, g) : \Gamma(E)} 
		\qquad 
		\inference[\rulename{PNeg}]{ \Gamma \vdash P : \tau }{\Gamma \vdash \neg P : \tau}
		\qquad
		\inference[\rulename{PAnd}]{\Gamma \vdash P_1 : \tau_1 \quad \Gamma \vdash P_2 : \tau_2}{\Gamma \vdash P_1 \land P_2 : \tau_1 \sqcup \tau_2}
		\qquad
		\inference[\rulename{POr}]{\Gamma \vdash P_1 : \tau_1 \quad \Gamma \vdash P_2 : \tau_2}{\Gamma \vdash P_1 \lor P_2 : \tau_1 \sqcup \tau_2}
		\\
		\inference[\rulename{Comp}]{\Gamma;\Theta \vdash p_1 : \tau_1 \quad \Gamma;\Theta \vdash p_2 : \tau_2}{\Gamma;\Theta \vdash p_1 \land p_2 : \tau_1 \sqcup \tau_2}
		\qquad 
		\inference[\rulename{Forbid}]{ \Gamma \vdash P : \tau }{ \Gamma;\Theta \vdash \forbid\ P : \tau }
		\qquad
		\inference[\rulename{ForbidWhen}]{\Gamma \vdash P : \tau \quad \tau_H = \Theta(H)}{\Gamma;\Theta \vdash \forbid\ P\ \when\ H : \tau \sqcup \tau_H}
		\\
		\inference[\rulename{Response}]{ \Gamma \vdash P_1 : \tau_1 \quad \Gamma \vdash P_2 : \tau_2}{\Gamma;\Theta \vdash \when\ P_1\ \within\ \Delta\ P_2 : T}
		\qquad
		\inference[\rulename{DEmpty}]{ }{\Gamma; \Theta \vdash \epsilon : \Theta} 
		\qquad 
		\inference[\rulename{DLet1}]{ H \notin \mathsf{dom}(\Theta) \\ \Gamma \vdash P : \tau \quad  \Gamma;\Theta[H \mapsto \tau]\vdash d^* : \Theta'}{\Gamma; \Theta \vdash \left(\mlet\ H := \happened(P)\right) d^* : \Theta'}
		\\
		\inference[\rulename{DLet2}]{ H \notin \mathsf{dom}(\Theta) \quad H' \in \mathsf{dom}(\Theta) \\ \Gamma \vdash P : \tau \quad \Gamma;\Theta[H \mapsto \tau \sqcup \Theta(H')]\vdash d^* : \Theta'}{\Gamma; \Theta \vdash \left(\mlet\ H := \happened(P)\ \when\ H'\right) d^* : \Theta'}
		\qquad
		\inference[\rulename{Prog}]{\Gamma;\emptyset \vdash d^* : \Theta \quad \Gamma; \Theta \vdash p : \tau \quad \Allowed(\tau, A)}{\Gamma \vdash \langle d^*\, p,\; A \rangle : \checkmark}
	\end{gather*}
	\caption{Typing rules for \thetool DSL.}\label{fig:full:typing}	
\end{figure*}

\subsubsection{Type System}\label{sec:app:typing}
\Cref{fig:full:typing} shows all the typing rules of our type system.
We briefly comment on the rules not already discussed in \Cref{sec:types}.
Rule \rulename{PNeg} obtains the type of $P$ and returns it since the negation operator does not change the way a violation may occur.
Rule \rulename{POr} is similar to the rule for conjunction: it takes the join of the types of the two sub-terms so as to ensure that a violation can occur in the worst way possible.

The typing rules for historical-predicate declarations are characterized by the judgment $\Gamma;\Theta \vdash d^* : \Theta'$,
where $\Theta$ is the violation-type environment accumulated so far and $\Theta'$ is the environment produced after processing $d^*$.
Rule \rulename{DEmpty} returns $\Theta$ unchanged when the declaration sequence is empty.
Rule \rulename{DLet1} derives the type $\tau$ of $P$ under $\Gamma$, extends $\Theta$ with the
binding $H \mapsto \tau$, and processes the remaining definitions $d^*$ in the extended environment, producing $\Theta'$. 
Rule \rulename{DLet2} follows the same approach for conditional declarations: 
it requires $H' \in \mathsf{dom}(\Theta)$ to ensure that $H'$ has already been declared and to guarantee
acyclicity. It extends $\Theta$ with $H \mapsto \tau \sqcup \Theta(H')$, where the join reflects that the violation type of
$H$ depends on both the predicate $P$ and the historical condition $H'$.  
Both rules \rulename{DLet1} and \rulename{DLet2} require $H \notin \mathsf{dom}(\Theta)$ to prevent redefinition of identifiers.
Note that both judgments only assign types in $\{C, O\}$: the type of an event predicate is either taken from $\Gamma$, whose codomain is $\{C, O\}$, or obtained by joining such types, and $\Theta$ is built from event-predicate types in the same way.
The type $T$ is therefore introduced only by rule \rulename{Response}.

Before proving type soundness, we prove an auxiliary result:
\begin{lemma}\label{lem:type:pred}
	Let $e$ be an event; if $\Gamma \vdash P : \tau$ and $\pred{P}(e,\,\emptyset) \neq \mnoapp$, then 
	$\Gamma(e.\mathit{type}) \sqsubseteq \tau$.
	\begin{proof}
	By induction on the derivation $\Gamma \vdash P : \tau$ and then by cases on the last rule applied.
	
	\medskip
	\noindent\emph{Case} \rulename{PEvent}: We have that $P = (E, g)$ and $\tau = \Gamma(E)$. 
	By the semantics we have $\pred{(E,g)}(e, \emptyset) \neq \mnoapp$ implies $e.\mathit{type} = E$. Therefore $\Gamma(E) = \Gamma(e.\mathit{type}) = \tau$.
	
	\medskip
	\noindent\emph{Case} \rulename{PNeg}: We have that $P = \lnot P'$, $\Gamma \vdash P' : \tau$. 
	By the semantics of negation, $\pred{\lnot P}(e, \emptyset) \neq \mnoapp$ implies $\pred{P'}(e, \emptyset) \neq \mnoapp$. By the induction hypothesis applied to $P'$, $\Gamma(e.\mathit{type}) \sqsubseteq \tau$.
	
	\medskip
	\noindent\emph{Case} \rulename{PAnd}: We have that $P = P_1 \land P_2$, $\tau = \tau_1 \sqcup \tau_2$, $\Gamma \vdash P_1 : \tau_1$, and $\Gamma \vdash P_2 : \tau_2$. 
	By hypothesis, 	$\pred{P_1 \land P_2}(e, \emptyset) \neq \mnoapp$, and we proceed by cases on $\pred{P_1}(e, \emptyset)$.
	If $\pred{P_1}(e, \emptyset) \neq \mnoapp$, we can apply the induction hypothesis to $P_1$, obtaining
	$\Gamma(e.\mathit{type}) \sqsubseteq \tau_1 \sqsubseteq \tau$.
	Otherwise, if $\pred{P_1}(e, \emptyset) = \mnoapp$, by the semantics of conjunction, the hypothesis
	$\pred{P_1 \land P_2}(e, \emptyset) \neq \mnoapp$ implies $\pred{P_2}(e, \emptyset) = \mfalse \neq \mnoapp$. By the induction hypothesis applied to $P_2$,
	$\Gamma(e.\mathit{type}) \sqsubseteq \tau_2 \sqsubseteq \tau$.
	In all cases $\Gamma(e.\mathit{type}) \sqsubseteq \tau$.
	
	\medskip
	\noindent\emph{Case} \rulename{POr}: We have that $P = P_1 \lor P_2$, $\tau = \tau_1 \sqcup \tau_2$, $\Gamma \vdash P_1 : \tau_1$, and $\Gamma \vdash P_2 : \tau_2$. 
	By hypothesis, $\pred{P_1 \lor P_2}(e, \emptyset) \neq \mnoapp$, and we proceed by cases on $\pred{P_1}(e, \emptyset)$.
	If $\pred{P_1}(e, \emptyset) = \mtrue \neq \mnoapp$, by the induction hypothesis applied to $P_1$, we have $\Gamma(e.\mathit{type}) \sqsubseteq \tau_1 \sqsubseteq \tau$.
	Otherwise, if $\pred{P_1}(e, \emptyset) = \mfalse$ or $\pred{P_1}(e, \emptyset) = \mnoapp$, by the semantics of disjunction, the hypothesis $\pred{P_1 \lor P_2}(e,
	\emptyset) \neq \mnoapp$ implies $\pred{P_2}(e, \emptyset) \neq \mnoapp$. 
	By the induction hypothesis applied to $P_2$, $\Gamma(e.\mathit{type}) \sqsubseteq \tau_2 \sqsubseteq \tau$.
	 In all cases $\Gamma(e.\mathit{type}) \sqsubseteq \tau$.
	\end{proof}
\end{lemma}
Finally, we prove \Cref{thm:type:cor}:
\typecor*
\begin{proof}
Let $\phi = \dsl{p}^D$ be the formula denoted by $p$, where $D$ is the definition environment obtained from $d^*$. Let $\sigma_{i^*}$ be the first violating prefix of $\sigma_k$. We proceed by induction on the derivation $\Gamma;\Theta \vdash p : \tau$ and by cases on the last rule applied.

\medskip
\noindent\emph{Case} \rulename{Forbid}: We know that $p = \texttt{forbid}\ P$, $\phi = G(\lnot P)$, $\tau = \tau_P$, and $\Gamma \vdash P : \tau_P$. 
Since $P$ is an event predicate, $\tau_P \in \{C,O\}$ by the predicate typing rules, hence $\tau \in \{C, O\}$.
By the semantics, $\sem{\pol}(\sigma_{i^*}) = \violated$ implies $\pi_1(\sem{\lnot P}(\sigma_{i^*}, \emptyset)) =
\violated$, that in turn implies $\pred{P}(e_{i^*}, \emptyset) = \mfalse$. Since $\pred{P}(e_{i^*}, \emptyset) \neq \mnoapp$, by \Cref{lem:type:pred} and since $\tau_P \in \{C,O\}$, $\Gamma(e_{i^*}.\mathit{type}) \sqsubseteq \tau_P = \tau$.

\medskip
\noindent\emph{Case} \rulename{ForbidWhen}: $p = \texttt{forbid}\ P\ \texttt{when}\ H$, $\phi = G(\lnot(D(H) \land P))$ where $D(H) = {<}\phi_H$ for some $\phi_H$, $\tau = \Theta(H) \sqcup \tau_P$,  $\Gamma \vdash P : \tau_P$ and $\Theta(H) = \tau_H$.
By the semantics, $\sem{\pol}(\sigma_{i^*}) = \violated$ implies $\pi_1(\sem{\lnot(D(H) \land P)}(\sigma_{i^*},
\emptyset)) = \violated$. By the semantics of negation, this implies $\pi_1(\sem{D(H) \land P}(\sigma_{i^*}, \emptyset)) = \ok$, that in turn implies $\pi_1(\sem{D(H)}(\sigma_{i^*}, \emptyset)) = \ok$ and $\pi_1(\sem{P}(\sigma_{i^*}, \emptyset)) = \ok$. 
The first condition means that $D(H) = {<}\phi_H$  holds on $\sigma_{i^*}$, meaning
there exists $j \leq i^*$ such that $\phi_H$ was satisfied at some past position $j$.
This means that a relevant past event occurred, but it does not identify the offending event.
The second condition $\pred{P}(e_{i^*}, \emptyset) = \mtrue$ identifies the offending event because 
it is satisfied by the current event $e_{i^*}$, which is the event that triggers the
violation at the first violating prefix.
Since $\pred{P}(e_{i^*}, \emptyset) = \mtrue \neq \mnoapp$, by \Cref{lem:type:pred} applied to
$P$, $\Gamma(e_{i^*}.\mathit{type}) \sqsubseteq \tau_P$.
Since $\tau_P \sqsubseteq \tau_H \sqcup \tau_P = \tau$, we conclude $\Gamma(e_{i^*}.\mathit{type}) \sqsubseteq \tau$; moreover $\tau_P, \tau_H \in \{C,O\}$, hence $\tau \in \{C,O\}$. 

\medskip
\noindent\emph{Case} \rulename{Response}: We know that $p = \when\ P_1\ \within\ \Delta\ P_2$, $\phi = G(P_1
\Fbound P_2)$, and $\tau = T$.
Since the resulting type is $\tau = T$, it suffices to establish the second disjunct of the thesis, namely that the violation is due to a deadline expiration. 
By the semantics, we have that $\sem{\pol}(\sigma_{i^*}) = \violated$ implies $\pi_1(\sem{P_1 \Fbound P_2}(\sigma_{i^*}, \emptyset)) = \violated$.
By the semantics of the operator $\Fbound$, this means there exists $(i, \rho_i, d_i) \in \mathcal{I}(\sigma_{i^*}, \emptyset)$ such that $e_{i^*}.t > d_i$ and $\neg\,\sat(i, \rho_i, d_i,
\sigma_{i^*})$, which is exactly the thesis, with $P_1 \Fbound P_2$ itself as the witness sub-formula.

\medskip
\noindent\emph{Case} \rulename{Comp}: We have that $p = p_1 \land p_2$, $\phi = \phi_1 \land \phi_2$, $\tau = \tau_1 \sqcup \tau_2$,  $\Gamma;\Theta \vdash p_1 : \tau_1$ and $\Gamma;\Theta \vdash p_2 : \tau_2$.
By a straightforward auxiliary induction on the syntax of $p$ we have $\bnd(\dsl{p}^D) = \emptyset$: the three base cases are $G$-formulas and $\bnd(G\,\psi) = \emptyset$ by definition, while $\bnd(\phi_1 \land \phi_2) = \bnd(\phi_1) \cup \bnd(\phi_2)$.
Hence, \emph{both} conjuncts are evaluated under the empty environment.
By the semantics of conjunction, $\sem{\pol}(\sigma_{i^*}) = \violated$ implies $\pi_1(\sem{\phi_1}(\sigma_{i^*},
\emptyset)) = \violated$ or $\pi_1(\sem{\phi_2}(\sigma_{i^*}, \emptyset)) = \violated$. 
Consider the case where $\phi_1$ is violated; the other one is symmetric.
Then $\sigma_{i^*}$ is also the \emph{first} violating prefix of $p_1$: if $\phi_1$ were violated at some $j < i^*$, the same clause for conjunction would give $\sem{\pol}(\sigma_j) = \violated$, contradicting the minimality of $i^*$.
The induction hypothesis for $p_1$ thus applies at $\sigma_{i^*}$, and we distinguish two cases on $\tau_1$.

If $\tau_1 \sqsubseteq O$, the induction hypothesis is the first item of the thesis for $p_1$: there is an offending event $e_{i^*}$ with $\Gamma(e_{i^*}.\mathit{type}) \sqsubseteq \tau_1$, and the same event is offending for $p_1 \land p_2$.
If $\tau \sqsubseteq O$, we must prove the first item, and $\Gamma(e_{i^*}.\mathit{type}) \sqsubseteq \tau_1 \sqsubseteq \tau$.
If $\tau = T$, we must prove the second item, and $\Gamma(e_{i^*}.\mathit{type}) \sqsubseteq \tau_1 \sqsubseteq O$ gives its first possibility.

If instead $\tau_1 = T$, then $\tau = T \sqcup \tau_2 = T$, so we must prove the second item of the thesis.
Since $\tau_1 = T$, the induction hypothesis is the second item of the thesis for $p_1$.
It gives us one of two possibilities.
Either there is an offending event $e_{i^*}$ for $p_1$ with $\Gamma(e_{i^*}.\mathit{type}) \sqsubseteq O$, and the same event is offending for $p_1 \land p_2$.
Or there exists a bounded-response sub-formula $P_1 \Fbound P_2$ of $\dsl{p_1}^D$ and an instance $(i, \rho_i, d_i) \in \mathcal{I}(\sigma_{i^*}, \emptyset)$ such that $e_{i^*}.t > d_i$ and $\neg\,\sat(i, \rho_i, d_i, \sigma_{i^*})$.
Since $\dsl{p}^D = \dsl{p_1}^D \land \dsl{p_2}^D$, the same sub-formula is one of $\dsl{p}^D$, and the same instance witnesses the thesis.
\end{proof}

\subsection{Monitor Synthesis}
\subsubsection{Normal Form}\label{sec:app:norm:form}
Given a DSL program $\pol = \langle d^*\,p, A\rangle$ and the definition environment $D$ obtained from $d^*$, we define our normalization function $\NF$ as follows:
\begin{align*}
	\NF(\langle d^*\,p, A\rangle) \;=\;& \NF(p) \\
	\NF(\forbid\ P) \;=\; & G(\lnot P)\\
	\NF(\forbid\ P\ \when\ H) \;=\; & G(\lnot (D(H) \land P))\\
	\NF(\when\ P_1\ \within\ \Delta\ P_2) \;=\; & G(P_1 \Fbound P_2)\\
	\NF(p_1 \land p_2) \;=\; & G(\phi'_1 \land \phi'_2)
\end{align*}
where $\NF(p_1) = G(\phi'_1)$, and $\NF(p_2) = G(\phi'_2)$.
Note that for every $p$, $\NF(p) = G(\phi')$ for a unique formula $\phi'$.
Hereafter, we denote by ${<}\phi_H$ the past formula $D(H)$ associated with the historical predicate $H$ in the definition environment $D$.

The following result ensures that our normalization process preserves the semantics of a DSL program:
\begin{theorem}[Semantic preservation]
	\label{th:nf-sem}
For every program $\pol = \langle d^*\,p, A \rangle$,
let $\phi = \dsl{p}^{D(d^*)}$ 
be the formula denoted by $\pol$. Then for every finite trace
$\sigma_k$:
\[
\pi_1(\sem{\phi}(\sigma_k, \emptyset))
\;=\;
\pi_1(\sem{\NF(p)}(\sigma_k, \emptyset)).
\]
\end{theorem}
\begin{proof}
We proceed by structural induction on the syntax of $p$, showing that $\pi_1(\sem{\phi}(\sigma_k, \emptyset)) = \pi_1(\sem{\NF(p)}(\sigma_k, \emptyset))$ for every finite trace $\sigma_k$.
	
	\medskip
	\noindent\emph{Case} \forbid\ $P$: 
	By the definition of $\dsl{\cdot}^D$ we have $\phi = G(\lnot P)$; since $\NF(p) = G(\lnot P)$,
	$\phi$ coincides syntactically with the normal form and the thesis holds trivially.
	
	\medskip
	\noindent\emph{Case} \forbid\ $P$ \when\ $H$:
	By the definition of $\dsl{\cdot}^D$ we have $\phi = G(\lnot(D(H) \land P))$, which is exactly
	$\NF(p)$; hence $\phi$ coincides syntactically with the normal form.
	
	\medskip
	\noindent\emph{Case} \when\ $P_1$ \within\ $\Delta$ $P_2$: 
	By the definition of $\dsl{\cdot}^D$ we have $\phi = G(P_1 \Fbound P_2)$; since $\NF(p) = \phi$,
	$\phi$ coincides syntactically with the normal form and the thesis holds trivially.
	
	\medskip
	\noindent\emph{Case} $p_1 \land p_2$: By the definition of $\dsl{\cdot}^D$ we have
	$\phi = \phi_1 \land \phi_2$, where $\phi_i = \dsl{p_i}^{D(d^*)}$ for $i \in \{1,2\}$; 
	by the definition of $\NF$ on conjunctions, we have $\NF(p) = G(\phi'_1 \land \phi'_2)$
	where $\phi'_i$ is such that $\NF(p_i) = G(\phi'_i)$.
	By the definition, $\dsl{\cdot}^D$ produces a $G$-formula or a conjunction of $G$-formulas, and
	by the definition of $\bnd$ on $G$-formulas we have $\bnd(\phi_i) = \emptyset$, and
	$\bnd(G(\phi'_i)) = \emptyset$ as well.
	By \Cref{prop:env-inv}, the evaluation of $\phi_1$ and of $G(\phi'_1)$ leaves the environment unchanged, so in both $\sem{\phi_1 \land \phi_2}(\sigma_k, \emptyset)$ and $\sem{G(\phi'_1) \land G(\phi'_2)}(\sigma_k, \emptyset)$ the second conjunct is evaluated under $\emptyset$ as well. Hence, by the semantics of conjunction, the value of each of the two formulas is determined by the values of its conjuncts. By the induction hypothesis and the definition of $\phi'_i$, for $i \in \{1,2\}$:
	\[
	\pi_1(\sem{\phi_i}(\sigma_k, \emptyset)) =
	\pi_1(\sem{\NF(p_i)}(\sigma_k, \emptyset)) =
	\pi_1(\sem{G(\phi'_i)}(\sigma_k, \emptyset)).
	\]
	Hence:
	\begin{align*}
		\pi_1(\sem{\phi}(\sigma_k, \emptyset))
		&= \pi_1(\sem{G(\phi'_1) \land G(\phi'_2)}(\sigma_k,
		\emptyset)) \\
		&= \pi_1(\sem{G(\phi'_1 \land \phi'_2)}(\sigma_k,
		\emptyset)) \\
		&= \pi_1(\sem{\NF(p)}(\sigma_k, \emptyset))
	\end{align*}
	where the first step follows from the argument above, the second from \Cref{prop:g-conj}, and
	the third from the definition of $\NF$ on conjunctions. The hypothesis of \Cref{prop:g-conj} holds: $\phi'_1$ is built
	from bounded-response formulas and negations, and neither introduces bindings, so
	$\bnd(\phi'_1) = \emptyset$.
\end{proof}

\subsubsection{Safety Monitor}\label{sec:app:smon}
We formally define the evaluation function $\eval{\psi}{\hm}{e} \in K_3$ inductively on the structure of the conjunct $\psi$ that can be either $P$ or ${<}\phi_H \land P$:
\begin{align*}
	\eval{P}{\hm}{e} \;=\;& \pi_1(\pred{P}(e, \emptyset)) \\
	\eval{{<}\phi_H \land P}{\hm}{e} \;=\; & 
	\begin{cases}
		\mfalse & \hm({<}\phi_H) = \pending \\
		\pi_1(\pred{P}(e, \emptyset)) & \text{otherwise}
	\end{cases}
\end{align*}

\begin{lemma}[Safety evaluation soundness]\label{lem:eval-sound}
	Given a DSL program $\pol$, for every safety conjunct $\psi_i$ in $\NF(\pol)$, finite trace $\sigma_k$ with $k \geq 1$, and historical memory $\hm$ such that for all ${<}\phi_H \in
	\pastF(\pol)$:
	\[
	\hm({<}\phi_H) = \ok \;\implies\; \pi_1(\sem{{<}\phi_H}(\sigma_k, \emptyset)) = \ok,
	\]
	we have:
	\[
	\eval{\psi_i}{\hm}{e_k} = \mtrue \;\implies\; \pi_1(\sem{\psi_i}(\sigma_k, \emptyset)) = \ok.
	\]
\end{lemma}
\begin{proof}
	By cases on the structure of $\psi_i$.
	
	\medskip
	\noindent\emph{Case $\psi_i = P$}: 
	by definition of $\eval{\cdot}{\cdot}{\cdot}$, $\eval{P}{\hm}{e_k} = \pred{P}(e_k, \emptyset) = \mtrue$. 
	By the semantics of predicates as formulas, $\pred{P}(e_k, \emptyset) = \mtrue$ implies $\pi_1(\sem{P}(\sigma_k, \emptyset))
	= \ok$. 
	
	\medskip
	\noindent\emph{Case $\psi_i = {<}\phi_H \land P$}: 
	By definition, $\eval{{<}\phi_H \land P}{\hm}{e_k} = \mtrue$ requires
	$\hm({<}\phi_H) = \ok$ and $\pred{P}(e_k,\emptyset) = \mtrue$. 
	By the hypothesis on $\hm$, $\hm({<}\phi_H) = \ok$ implies $\pi_1(\sem{{<}\phi_H}(\sigma_k, \emptyset)) = \ok$.
	By the semantics of predicates as formulas, $\pred{P}(e_k, \emptyset) = \mtrue$ implies
	$\pi_1(\sem{P}(\sigma_k, \emptyset)) = \ok$. 
	By the semantics of conjunction, we conclude:
	\[
	\pi_1(\sem{{<}\phi_H \land P}(\sigma_k,
	\emptyset)) = \ok.
	\]
\end{proof}

\begin{lemma}[Safety evaluation completeness]\label{lem:eval-compl}
	Given a DSL program $\pol$, for every safety conjunct $\psi_i$ in $\NF(\pol)$,
	finite trace $\sigma_k$ with $k \geq 1$, and historical memory $\hm$ such that for all ${<}\phi_H \in
	\pastF(\pol)$:
	\[
	\pi_1(\sem{{<}\phi_H}(\sigma_k, \emptyset)) = \ok
	\;\implies\;
	\hm({<}\phi_H) = \ok,
	\]
	if $\pi_1(\sem{\psi_i}(\sigma_k, \emptyset)) = \ok$
	and $\pred{P}(e_k, \emptyset) \neq \mnoapp$, where $P$ is the event predicate occurring in
	$\psi_i$, then:
	\[
	\eval{\psi_i}{\hm}{e_k} = \mtrue.
	\]
\end{lemma}
\begin{proof}
	By cases on the structure of $\psi_i$.
	
	\medskip
	\noindent\emph{Case $\psi_i = P$}:
	By the semantics of predicates as formulas, $\pi_1(\sem{P}(\sigma_k, \emptyset)) = \ok$ means
	$\pred{P}(e_k, \emptyset) = \mtrue$, hence by
	definition of $\eval{\cdot}{\cdot}{\cdot}$:
	\[
	\eval{P}{\hm}{e_k} = \mtrue.
	\]
	
	\medskip
	\noindent\emph{Case $\psi_i = {<}\phi_H \land P$}:
	$\pi_1(\sem{{<}\phi_H \land P}(\sigma_k, \emptyset)) = \ok$ implies by the semantics of conjunction both $\pi_1(\sem{{<}\phi_H}(\sigma_k, \emptyset)) = \ok$
	and $\pi_1(\sem{P}(\sigma_k, \emptyset)) = \ok$.
	By the hypothesis on $\hm$, $\pi_1(\sem{{<}\phi_H}(\sigma_k, \emptyset)) = \ok$
	implies $\hm({<}\phi_H) = \ok$.
	By the same argument as the case $\psi_i = P$, $\pi_1(\sem{P}(\sigma_k, \emptyset)) = \ok$
	implies $\pred{P}(e_k, \emptyset) = \mtrue$.
	
	Therefore by definition of $\eval{\cdot}{\cdot}{\cdot}$:
	\[
	\eval{{<}\phi_H \land P}{\hm}{e_k} = \mtrue.
	\]
\end{proof}

\subsubsection{Bounded-response Monitor}\label{sec:app:bresp}
We define the memory update function $U$ used in the bounded-response monitor automaton.
Consider a formula $P^j_1 \Fboundj P^j_2$.
Given an active instance set $I_j$,
a violation flag $f_j$, and an event $e$, we define:
\begin{align*}
	\mathsf{vi}(I_j, e)
	&\iff \exists\,(d,\rho') \in I_j.\;
	e.t > d 
	\\[4pt]
	\mathsf{rm}(I_j, e)
	&= I_j \setminus \{(d,\rho') \mid
	e.t > d \;\lor\;
	(e.t \leq d \;\land\;
	\pred{P^j_2}(e,\rho') = \mtrue)\}
	\\[4pt]
	\mathsf{new}(e)
	&= \begin{cases}
		\{(e.t + \Delta_j, \rho')\} &
		\text{if } \pred{P^j_1}(e,\emptyset) = (\mtrue, \rho') \\
		\emptyset & \text{otherwise}
	\end{cases}
\end{align*}
where the predicate $\mathsf{vi}(I_j, e)$ detects whether any active instance
has expired. 
The function $\mathsf{rm}(I_j, e)$ removes all instances that are either expired or satisfied by $e$. 
The function $\mathsf{new}(e)$ produces a new instance if
$P^j_1$ is triggered by $e$. The order of application
is fixed: $\mathsf{vi}$ is evaluated on the current $I_j$
before $\mathsf{rm}$ and $\mathsf{new}$ update it.

The memory update function is defined as $U((I_j, f_j), e) = (I'_j, f'_j)$
where:
\[
I'_j = \mathsf{rm}(I_j, e) \cup \mathsf{new}(e), \quad
f'_j = \begin{cases}
	\violated & \text{if } f_j = \violated \;\lor\;
	\mathsf{vi}(I_j, e) \\
	\ok & \text{otherwise.}
\end{cases}
\]

\subsubsection{Global Monitor}\label{sec:app:glob:mon}
Consider a DSL program $\pol$ and its normal form formula $\NF(\pol)$ that has safety and bounded response conjuncts indexed by the sets $I$ and $J$, respectively.
Recall that the states of the monitor $\mathcal{A}=\GenMon(\NF(\pol))$ are $(|I| + |J|)$ tuples, whose elements track the corresponding state of each local automaton. 
Given a state $q$ of $\mathcal{A}$, we denote by $q_k$ the $k$-th element of $q$.  
Recall that we denote by ${<}\phi_H$ the past formula $D(H)$ associated with the historical predicate $H$ in the definition environment $D$ built from the declarations $d^*$ of $\pol$.

Before presenting the rule for the global monitor we formalize the operation to update the historical memory.
More precisely, we define $\updH$ by recursion on the declaration sequence $d^*$ occurring inside a DSL program $\pol$:
\begin{align*}
&\updH(\hm, e, \epsilon) = \hm \\
&\updH(\hm, e, d^* \cdot \mlet\ H := \happened(P))
= \\
& \qquad \qquad \begin{cases}
	\hm'[{<}P \mapsto \ok] &
	\text{if } \pred{P}(e, \emptyset) = \mtrue \\
	\hm' &
	\text{otherwise}
\end{cases}
\\
& \updH(\hm, e, d^* \cdot \mlet\ H :=
\happened(P)\ \when\ H')
= \\
& \qquad \begin{cases}
	\hm'[{<}(D(H') \land P) \mapsto \ok] &
	\text{if } \hm'({<}\phi_{H'}) = \ok \;\land\; \\
	& \quad \pred{P}(e, \emptyset) = \mtrue \\
	\hm' &
	\text{otherwise}
\end{cases}
\end{align*}
where $\hm' = \updH(\hm, e, d^*)$ is the memory already updated for all preceding
declarations, according to the order in which the declarations occur in $d^*$. 
Since we know that declarations $d^*$ are acyclic by the type system, the updates follow a topological order.

Rule \rulename{Par} below captures the behavior of the global monitor: 
\[
\inference[\rulename{Par}]{
	\hm' = \updH(\hm, e, d^*)
	\\
	\forall i \in I.\;
	(q_i,\, (\hm', \jm)) \xrightarrow{\,e\,}^s_i (q'_i,\, (\hm',\jm))
	\\
	\forall j \in J.\;
	(q_j,\, (\hm', \jm)) \xrightarrow{\,e\,}^r_j (q'_j,\, (\hm',\jm_j))
	\\
	\jm' = \jm\bigl[j \mapsto \jm_j(j) \mid j \in J \bigr]
}{
	\left(q, (\hm,\jm)\right) \;\xrightarrow{\,e\,}\; \left(q', (\hm',\jm')\right)
}
\]
The first premise of \rulename{Par} updates the historical memory $\hm$ to $\hm'$ by setting to $\ok$ every entry ${<}\phi_H$ whose formula $\phi_H$ evaluates to $\mtrue$ on the current event $e$, leaving all other entries unchanged. We update the memory following the declaration list $d^*$ of the program $\pol$.
This update is monotone: entries already set to $\ok$ are never reset.
The second and the third premises fire the corresponding transition on the local automaton associated with the conjunct of the formula. 
The last premise changes the clause state of each bounded response $j$ to update the second component of the monitor memory. 
Since the updates of each bounded-response automaton are independent, we can simply merge the various updates to obtain $\jm'$.  

\begin{proposition}[Determinism]\label{prop:determinism}
For every DSL program $\pol$ and finite trace $\sigma_k$, there exists a unique run
\[
	(q_0, m_0) \xrightarrow{\sigma_k} (q, m).
\]
\end{proposition}
\begin{proof}
	It suffices to show that for every configuration $(q, m)$ and event $e$, there exists a unique
	configuration $(q', m')$ such that $(q, m) \xrightarrow{e} (q', m')$. 
	The result then follows by induction on $k$, with the base case given by
	the unique initial configuration $(q_0, m_0)$.
	
	We reason on how a transition is performed by rule \rulename{Par}: we show that every step in the premise of rule \rulename{Par} produces a unique output given a configuration $(q,m)$ and an event $e$.
	
	First, given $e$ and $m=(\hm, \jm)$ the next historical memory is obtained as
	$\hm' = \updH(\hm, e, d^*)$. 
	By definition $\updH$ is a function defined by structural recursion on $d^*$, so $\hm'$ is uniquely determined. 
	
	Consider now the transition of safety automata. Recall that $q$ is a tuple where each element records the state $q_i$ of a local automaton.  
	For each $i \in I$, the new control state $q'_i$ is uniquely determined by $q_i$ and by the evaluation of
	$\eval{\psi_i}{\hm'}{e}$. 
	There are two cases according to the state $q_i$. 
	If $q_i = \violated$: only rule \rulename{S-Absorb} applies, giving $q'_i = \violated$.
	Otherwise, if $q_i = \ok$, then exactly one of rules \rulename{S-Viol} or \rulename{S-Ok} applies, depending on whether $\eval{\psi_i}{\hm'}{e} = \mtrue$ or not. These two rules are mutually exclusive and exhaustive. 
		
	Consider now the transition of the bounded-response automata. 
	For each $j \in J$, the function $U((I_j, f_j), e)$ is defined as
	a composition of $\mathsf{vi}$, $\mathsf{rm}$, and $\mathsf{new}$, each of which is a function of its
	inputs. Hence $U$ produces a unique output $(I'_j, f'_j)$. 
	The new control state $q'_j$ is then uniquely determined as follows.
	If $q_j = \violated$, then only rule \rulename{R-Absorb} applies, giving $q'_j = \violated$.
	Otherwise, if $q_j = \ok$, then exactly one of rules \rulename{R-Viol}
	or \rulename{R-Ok} applies, depending on whether $f'_j = \violated$ or $f'_j = \ok$. 
	These two conditions are mutually exclusive and exhaustive. 
	
	Finally, observe that the component $J$ in $m$ is updated according to $\jm' = \jm[j \mapsto (I'_j, f'_j) \mid j \in J]$, which produces a unique output. 
	
	Therefore, every step by rule \rulename{Par} is deterministic and produces a unique configuration $(q', m')$ from $(q, m)$ and $e$.
\end{proof}

We now provide a bound on the size of the monitor: we show that the size of the monitor is linear in the size of the formula and in the number of historical-predicate declarations.
Note that in a real eBPF deployment we can only admit a finite and bounded number of active instances for each bounded response clause, denoted $\mathit{MAX\_INST}$. 
Let $\pol = \langle d^*\,p, A \rangle$ be a DSL program with normal form
\[
\NF(\pol) = G \Bigl(
\bigwedge_{i \in I} \lnot\psi_i
\;\land\;
\bigwedge_{j \in J}
\bigl(P^j_1\,\Fboundj\,P^j_2\bigr)
\Bigr).
\]
We use the following parameters throughout:
\begin{itemize}
	\item $|I|$: number of safety conjuncts;
	\item $|J|$: number of bounded-response conjuncts;
	\item $n = |d^*|$: number of historical-predicate declarations in $\pol$;
	\item $v_j = |\mathit{vars}(P^j_1)|$: number of correlation variables bound by $P^j_1$, for each $j \in J$;
	\item $w$: machine word size in bits;
	\item $\mathit{MAX\_INST} \in \mathbb{N}$: the maximum number of active instances per bounded-response
	clause, a fixed deployment parameter.
\end{itemize}

\begin{proposition}[Monitor size]\label{prop:monitor-size}
	Let $\pol$ be a DSL program and $\mathcal{A} = \GenMon(\NF(\pol))$ be its corresponding monitor. Then $\mathcal{A}$ requires:
	\begin{itemize}
		\item $|I| + |J|$ bits for the control state;
		\item $n$ bits for the historical memory;
		\item $|J|$ bits for the violation flags;
		\item at most $\mathit{MAX\_INST} \cdot (64 + v_j \cdot w)$ bits for the active instance
		set of each clause $j \in J$.
	\end{itemize}
	Its total size is:
	\[
	|I| + 2|J| + n \;+\; \mathit{MAX\_INST} \cdot \sum_{j \in J} (64 + v_j \cdot w) \quad\text{bits.}
	\]
\end{proposition}
\begin{proof}
	We account for each component of the monitor state.
	
	\medskip
	\noindent\emph{Control state.}
	Each local automaton $A^s_i$ and $A^r_j$ uses a single bit to represent its current state $q_i, q_j \in \{\ok, \violated\}$, for a total of $|I| + |J|$ bits. 
	
	\medskip
	\noindent\emph{Historical memory.}
	A historical memory $\hm$ maps each of the $n$ past formulas to a single
	bit, requiring $n$ bits in total. 
	
	\medskip
	\noindent\emph{Violation flags.}
	Each bounded-response clause $j \in J$ maintains a violation flag $f_j \in \{\ok, \violated\}$,
	one bit per clause, $|J|$ bits in total.

	\medskip
	\noindent\emph{Active instance sets.}
	Each clause $j \in J$ maintains at most $\mathit{MAX\_INST}$ active instances in a real deployment. 
	Each instance $(d, \rho')$ consists of:
	\begin{itemize}
		\item a deadline $d \in \mathbb{N}$, stored as a 64-bit nanosecond timestamp;
		\item a correlation environment $\rho'$ containing $v_j$ values, each stored as a
		$w$-bit machine word, for a total of $v_j \cdot w$ bits.
	\end{itemize}
	Each instance therefore requires $64 + v_j \cdot w$ bits, and the active instance set of clause
	$j$ requires at most $\mathit{MAX\_INST} \cdot (64 + v_j \cdot w)$ bits. 
	
	\medskip
	\noindent\emph{Total.}
	Summing all components:
	\[
	\underbrace{|I| + |J|}_{\text{control}}
	\;+\;
	\underbrace{n}_{\text{history}}
	\;+\;
	\underbrace{|J|}_{\text{flags}}
	\;+\;
	\underbrace{\mathit{MAX\_INST} \cdot
		\sum_{j \in J}(64 + v_j \cdot w)}_{\text{instances}}
	\quad\text{bits.}
	\]
	The first three terms give $|I| + 2|J| + n$ bits, which is $O(|I| + |J| + n)$ in the size of $\NF(\pol)$ and $d^*$. 
	The instance storage is linear in $\mathit{MAX\_INST} \cdot	|J| \cdot (64 + v_{\max} \cdot w)$ where
	$v_{\max} = \max_{j \in J} v_j$, which is $O(\mathit{MAX\_INST} \cdot |J|)$ for a fixed policy with bounded correlation variables.
\end{proof}

\subsubsection{Proof of Correctness}\label{sec:app:mcorr:proofs}
We first introduce some notation and two auxiliary results.
Given a DSL program $\pol$, the corresponding monitor $\mathcal{A} = \GenMon(\NF(\pol))$, and a finite trace $\sigma_k$,
we write $(q_0, m_0) \xrightarrow{\sigma_k} (q, m)$ to denote the run induced by the trace $\sigma_k$, i.e., a sequence of $k$ transitions from the initial configuration $(q_0, m_0)$ that consumes all events in $\sigma_k$ and reaches $(q, m)$. 

To prove soundness we need the following auxiliary lemmas:
\begin{lemma}[UpdH soundness]\label{lem:updh-step}
	For every event $e$, historical memory $\hm$, and declaration sequence $d^*$ of a DSL program $\pol$, if
	\[
	\forall\, {<}\phi_H \in \pastF(\pol).\;
	\hm({<}\phi_H) = \ok \;\implies\; \pi_1(\sem{{<}\phi_H}(\sigma, \emptyset)) = \ok
	\]
	for some finite trace $\sigma$, then:
	\[
	\updH(\hm, e, d^*)({<}\phi_H) = \ok \;\implies\;
	\pi_1(\sem{{<}\phi_H}(\sigma', \emptyset)) = \ok
	\]
	for any ${<}\phi_H \in \pastF(\pol)$ and finite trace $\sigma'$ that extends $\sigma$ with $e$ as its last event.
\end{lemma}
\begin{proof}
	By induction on $d^*$.
	
	\medskip
	\noindent\emph{Base case} ($d^* = \epsilon$):
	By definition we have $\updH(\hm, e, \epsilon) = \hm$, so the lemma directly follows from the hypothesis.
	
	\medskip
	\noindent\emph{Inductive step}: 
	let $d^* = d^*_0 \cdot d$ where $d$ is the last declaration. Let
	$\hm_0 = \updH(\hm, e, d^*_0)$. 
	By the inductive hypothesis applied to $d^*_0$, we have:
	\[
	\forall\, {<}\phi_H \in \pastF(\pol).\;
	\hm_0({<}\phi_H) = \ok
	\;\implies\;
	\pi_1(\sem{{<}\phi_H}(\sigma', \emptyset)) = \ok.
	\]
	We proceed by cases on the last declaration $d$.
	
	\medskip
	\noindent\emph{Case} $d = \mlet\ H := \happened(P)$: by definition of
	$\updH$:
	\[
	\updH(\hm, e, d^*) =
	\begin{cases}
		\hm_0[{<}P \mapsto \ok] &
		\text{if } \pred{P}(e, \emptyset) = \mtrue
		\\
		\hm_0 & \text{otherwise.}
	\end{cases}
	\]
	Suppose $\updH(\hm, e, d^*)({<} \phi_H) = \ok$ for some ${<} \phi_H$.
	For ${<} \phi_H \neq {<} P$ the value is unchanged from $\hm_0$ ($\hm_0({<}\phi_H) = \ok$), so by induction hypothesis, $\pi_1(\sem{{<}\phi_H}(\sigma', \emptyset)) = \ok$.
	For ${<} \phi_H = {<} P$ we have two cases. 
	If $\hm_0({<}\phi_H) = \ok$, then by induction hypothesis, $\pi_1(\sem{{<}\phi_H}(\sigma', \emptyset)) = \ok$. 
	Otherwise, $\hm_0({<}\phi_H) = \pending$ and $\pred{P}(e, \emptyset) = \mtrue$.
	By the semantics of predicates as formulas,
	$\pi_1(\sem{P}(\sigma', \emptyset)) = \ok$.
	Then, by the semantics of ${<}$ we have:
	\[
		\pi_1(\sem{{<}P}(\sigma', \emptyset)) = \ok.
	\]

	\medskip
	\noindent\emph{Case} $d = \mlet\ H := \happened(P)\ \when\ H'$: 
	by definition of $\updH$:
	\[
	\updH(\hm, e, d^*) =
	\begin{cases}
		\hm_0[{<}({<}\phi_{H'} \land P)
		\mapsto \ok] &
		\text{if } \hm_0({<}\phi_{H'})
		= \ok \;\land\; \\
		& 
		\pred{P}(e, \emptyset) = \mtrue \\
		\hm_0 & \text{otherwise.}
	\end{cases}
	\]
	Suppose $\updH(\hm, e, d^*)({<}
	\phi_H) = \ok$. For ${<}\phi_H \neq {<}({<}\phi_{H'}
	\land P)$, the value is unchanged from $\hm_0$,
	so the conclusion follows from induction hypothesis. 
	For ${<}\phi_H = {<}({<}\phi_{H'} \land P)$, we have two cases.
	If its value was already $\ok$, $\hm_0({<}({<}\phi_{H'} \land P)) = \ok$, the thesis follows by induction hypothesis, 
	$\pi_1(\sem{{<}({<}\phi_{H'} \land P)}(\sigma',
	\emptyset)) = \ok$. 
	If its value was not already set, i.e., $\hm_0({<}({<}\phi_{H'} \land P)) = \pending$ and $\hm_0({<}\phi_{H'}) = \ok$
	and $\pred{P}(e, \emptyset) = \mtrue$, we have that by induction hypothesis applied to ${<}\phi_{H'}$, $\hm_0({<}\phi_{H'}) = \ok$ implies
		$\pi_1(\sem{{<}\phi_{H'}}(\sigma', \emptyset)) = \ok$. 
	By the semantics of predicates as formulas,
	$\pred{P}(e, \emptyset) = \mtrue$ implies
	$\pi_1(\sem{P}(\sigma', \emptyset)) = \ok$. 
	By the semantics of conjunction:
	\[
		\pi_1(\sem{{<}\phi_{H'} \land P}(\sigma',\emptyset)) = \ok.
	\]
	and by the semantics of ${<}$, we have:
	\[
		\pi_1(\sem{{<}({<}\phi_{H'} \land P)}(\sigma',\emptyset)) = \ok. 
	\]
\end{proof}

\begin{lemma}[Historical memory soundness]\label{lem:hm-sound}
	Let $\pol$ be a DSL program, and let $(q_0, m_0) \xrightarrow{\sigma_k} (q, m)$ be the
	run of the monitor $\mathcal{A} = \GenMon(\NF(\pol))$ on $\sigma_k$.
	For every ${<}\phi_H \in \pastF(\pol)$:
	\[
	m_\hm({<}\phi_H) = \ok 	\;\implies\; \pi_1(\sem{{<}\phi_H}(\sigma_k, \emptyset)) = \ok.
	\]
\end{lemma}
\begin{proof}
	By induction on the length of the trace $\sigma_k$.
	
	\medskip
	\noindent\emph{Base case} ($k = 0$): the historical memory $\hm$ of the initial memory $m_0$ is such that $\hm({<}\phi_H) = \pending \neq \ok$ by definition, so the lemma is vacuously satisfied. 
	
	\medskip
	\noindent\emph{Inductive step}: assume the lemma holds for a trace $\sigma_{k-1}$, namely
	for the run $(q_0, m_0) \xrightarrow{\sigma_{k-1}} (q'', m'')$,  we have that for all ${<}\phi_H \in \pastF(\pol)$:
	\[
	m''_\hm({<}\phi_H) = \ok \;\implies\; \pi_1(\sem{{<}\phi_H}(\sigma_{k-1}, \emptyset)) = \ok.
	\]
	Consider an extension of $\sigma_{k-1}$ with a new event $e_k$, and consider the last transition of the run 
	$(q_0, m_0) \xrightarrow{\sigma_{k}} (q, m)$.
	By the premise of rule \rulename{Par}, 
	$m_\hm = \updH(m''_\hm, e_k, d^*)$.
	The inductive hypothesis is exactly the premise of \Cref{lem:updh-step} instantiated with
	$m''_\hm$ and $\sigma =	\sigma_{k-1}$. By applying \Cref{lem:updh-step},
	taking $\sigma' = \sigma_{k-1} \cdot e_k = \sigma_k$, we conclude: for all ${<}\phi_H \in
	\pastF(\pol)$:
	\[
	\updH(m''_\hm, e_k, d^*)({<}\phi_H) = \ok 	\;\implies\; \pi_1(\sem{{<}\phi_H}(\sigma_k, \emptyset)) = \ok,
	\]
	which is exactly:
	\[
	m_\hm({<}\phi_H) = \ok 	\;\implies\; \pi_1(\sem{{<}\phi_H}(\sigma_k, \emptyset)) = \ok. 
	\]
\end{proof}

\monsound*
\begin{proof}
	By induction on the length $k$ of the trace $\sigma_k$.
	
	\medskip
	\noindent\emph{Base case} ($k = 0$): the initial configuration $(q_0, m_0)$ is not bad by definition,
	so the hypothesis is vacuously satisfied. 
	
	\medskip
	\noindent\emph{Inductive step}: assume the theorem holds for a trace $\sigma_{k-1}$, i.e., if the monitor $\mathcal{A}$ reaches a bad configuration after reading $\sigma_{k-1}$
	then $\sem{\pol}(\sigma_{k-1}) = \violated$. 
	We show it holds for $\sigma_k$.
	
	Suppose $\mathcal{A}$ reaches a bad configuration after reading $\sigma_k$. 
	There are two cases:
	
	\medskip
	\noindent\emph{Case 1}: $\mathcal{A}$ was already in a bad configuration after reading $\sigma_{k-1}$. 
	By the inductive hypothesis, $\sem{\pol}(\sigma_{k-1}) = \violated$, and by \Cref{thm:persistence},
	$\sem{\pol}(\sigma_k) = \violated$. 
	
	\medskip
	\noindent\emph{Case 2}: $\mathcal{A}$ was not in a bad configuration after reading $\sigma_{k-1}$, but entered a bad configuration at step $k$. 
	Then at least one local automaton transitioned to $\violated$ at step $k$
	via a non-absorbing rule. We proceed by cases on which automaton.
	
	\medskip
	\noindent\emph{Safety case.} Some $A^s_i$ transitioned to $\violated$ via \rulename{S-Viol}, which requires:
	\[
	\eval{\psi_i}{m_\hm}{e_k} = \mtrue
	\]
	where $m_\hm$ is the historical memory after the update at step $k$ in the premise of rule \rulename{Par}. 
	By \Cref{lem:hm-sound} applied to the
	run on $\sigma_k$, $m_\hm({<}\phi_H) = \ok$ implies $\pi_1(\sem{{<}\phi_H}(\sigma_k, \emptyset))
	= \ok$ for all ${<}\phi_H \in \pastF(\pol)$. 
	By applying \Cref{lem:eval-sound} with $m_\hm$, we have:
	\[
	\pi_1(\sem{\psi_i}(\sigma_k, \emptyset)) = \ok,
	\]
	hence $\pi_1(\sem{\lnot\psi_i}(\sigma_k, \emptyset))
	= \violated$ by the semantics of $\lnot$. 
	Thus, since $\NF(\pol)$ is a conjunction under the operator $G$, by the semantics of $\land$ and $G$ we have:
	\[
	\pi_1(\sem{\NF(\pol)}(\sigma_k,
	\emptyset)) = \violated.
	\]
	
	\medskip
	\noindent\emph{Bounded-response case.}
	Some $A^r_j$ transitioned to $\violated$ via \rulename{R-Viol}, which requires the
	update function to return the flag $\violated$. Since $A^r_j$ was not in a bad
	configuration after $\sigma_{k-1}$, its flag was $\ok$. 
	A bounded-response automaton is in the $\violated$ state exactly when its flag is,  hence, by the
	definition of $U$ $\mathsf{vi}(I_j, e_k)$ holds, and there exists
	$(d, \rho') \in I_j$ with $e_k.t > d$.
	The instance $(d, \rho')$ was inserted by $\mathsf{new}$ at some step $i < k$, where
	$\pred{P^j_1}(e_i, \emptyset) = (\mtrue, \rho')$ and $d = e_i.t + \Delta_j$.
	By the semantics of predicates as formulas, $\pi_1(\sem{P^j_1}(\sigma_i, \emptyset)) = \ok$,
	which means $(i, \rho', d) \in \mathcal{I}(\sigma_k, \emptyset)$.
	Since $(d, \rho')$ is still in $I_j$ at step $k$, because $\mathsf{vi}$ is evaluated before
	$\mathsf{rm}$ updates the set, the instance was never removed at any step
	$\ell \in (i, k)$. By the definition of $\mathsf{rm}$, there is then no $\ell \in (i,k)$ with
	$e_\ell.t \leq d$ and $\pred{P^j_2}(e_\ell, \rho') = \mtrue$; since $e_k.t > d$, no such step
	exists in $(i,k]$ either, that is, $\lnot\sat(i, \rho', d, \sigma_k)$.
	By the semantics of the operator $\Fbound$:
	\[
	\pi_1(\sem{P^j_1\,\Fboundj\,P^j_2}(\sigma_k, \emptyset)) = \violated.
	\]
	Thus, since $\NF(\pol)$ is a conjunction under the operator $G$, by the semantics of $\land$ and $G$ we have:
	\[
	\pi_1(\sem{\NF(\pol)}(\sigma_k,
	\emptyset)) = \violated.
	\]
	
	\medskip
	\noindent Thus, independently of how the local monitor enters the bad state, by
	\Cref{th:nf-sem}:
	\[
	\sem{\pol}(\sigma_k) \;=\;
	\pi_1(\sem{\NF(\pol)}(\sigma_k, \emptyset))
	= \violated.
	\]
\end{proof}

To prove completeness, we need the following auxiliary lemmas:
\begin{lemma}[UpdH completeness]\label{lem:updh-compl}
For every event $e$, historical memory $\hm$, declaration sequence $d^*$ of a DSL program $\pol$, and finite trace $\sigma_k$, if
\[
\forall\, {<}\phi_H \in \pastF(\pol).\;
\pi_1(\sem{{<}\phi_H}(\sigma_k, \emptyset)) = \ok
\;\implies\;
\hm({<}\phi_H) = \ok,
\]
then:
\[
\pi_1(\sem{{<}\phi_H}(\sigma', \emptyset)) = \ok \;\implies\;
\updH(\hm, e, d^*)({<}\phi_H) = \ok,
\]
for any ${<}\phi_H \in \pastF(\pol)$ and $\sigma'=\sigma_k \cdot e$.
\end{lemma}
\begin{proof}
	We assume $\pi_1(\sem{{<}\phi_H}(\sigma', \emptyset)) = \ok$ and proceed by induction on $d^*$.
	
	\medskip
	\noindent\emph{Base case} ($d^* = \epsilon$): 
	$\updH(\hm, e, \epsilon) = \hm$, so the lemma directly holds. 
	
	\medskip
	\noindent\emph{Inductive step}: 
	Let $d^* = d^*_0 \cdot d$ where $d$ is the last declaration, and let $\hm_0 = \updH(\hm, e, d^*_0)$. 
	By the inductive hypothesis applied to $d^*_0$, we have:
	\[
	\forall\, {<}\phi_H \in \pastF(\pol).\;
	\pi_1(\sem{{<}\phi_H}(\sigma', \emptyset)) = \ok
	\;\implies\;
	\hm_0({<}\phi_H) = \ok.
	\]
	We proceed by cases on $d$.
	
	\medskip
	\noindent\emph{Case} $d = \mlet\ H := \happened(P)$: 
	by definition of $\updH$:
	\[
	\updH(\hm, e, d^*) =
	\begin{cases}
		\hm_0[{<}P \mapsto \ok] &
		\text{if } \pred{P}(e, \emptyset) = \mtrue
		\\
		\hm_0 & \text{otherwise.}
	\end{cases}
	\]
	Our hypothesis is $\pi_1(\sem{{<}\phi_H}(\sigma', \emptyset)) = \ok$ and $\sigma' = \sigma_{k} \cdot e$.
	For ${<}\phi_H \neq {<}P$, the conclusion follows from the induction hypothesis. 
	For ${<}\phi_H = {<}P$, by the semantics of ${<}$, $\exists\, i
	\leq k + 1$ such that $\pi_1(\sem{P}(\sigma_i, \emptyset)) = \ok$. 
	We have two cases. 
	If $i \leq k$, by the semantics of ${<}$ applied to $\sigma_k$,
	$\pi_1(\sem{{<}P}(\sigma_k, \emptyset)) = \ok$.
	By the hypothesis of the lemma, $\hm({<}P) = \ok$. By induction hypothesis,
	$\hm_0({<}P) = \ok$, hence $\updH(\hm, e, d^*)({<}P) = \ok$.
	Otherwise if $i = k+1$, then $\pi_1(\sem{P}(\sigma', \emptyset)) = \ok$, 
	hence $\pred{P}(e, \emptyset) = \mtrue$, and by definition of $\updH$ we have $\updH(\hm, e, d^*)({<}P) = \ok$.
	
	\medskip	
	\noindent\emph{Case} $d = \mlet\ H := \happened(P)\ \when\ H'$: 
	by definition of $\updH$:
	\[
	\updH(\hm, e, d^*) =
	\begin{cases}
		\hm_0[{<}({<}\phi_{H'} \land P) \mapsto \ok] &
		\text{if } \hm_0({<}\phi_{H'})
		= \ok \;\land\; \\
		& \pred{P}(e, \emptyset) = \mtrue \\
		\hm_0 & \text{otherwise.}
	\end{cases}
	\]
	Our hypothesis is $\pi_1(\sem{{<}\phi_H}(\sigma_k \cdot e, \emptyset)) = \ok$. 
	For ${<}\phi_H \neq {<}({<}\phi_{H'} \land P)$, the conclusion follows from induction hypothesis. 
	For ${<}\phi_H = {<}({<}\phi_{H'} \land P)$, by the semantics of ${<}$,
	$\exists\, i \leq k+1$ such that $\pi_1(\sem{{<}\phi_{H'} \land P}(\sigma_i, \emptyset)) = \ok$. 
	This could happen in two cases:
	If $i \leq k$, then  by the semantics of ${<}$ applied to $\sigma_k$,
	$\pi_1(\sem{{<}({<}\phi_{H'} \land P)}(\sigma_k, \emptyset)) = \ok$. 
	By the hypothesis of the lemma, $\hm({<}({<}\phi_{H'} \land P)) = \ok$. 
	By induction hypothesis, $\hm_0({<}({<}\phi_{H'} \land P)) = \ok$, 
	hence $\updH(\hm, e, d^*)({<}({<}\phi_{H'} \land P)) = \ok$.
	Otherwise if $i = k+1$, then $\pi_1(\sem{{<}\phi_{H'} \land P}(\sigma_k \cdot e, \emptyset)) = \ok$,
	which implies both $\pi_1(\sem{{<}\phi_{H'}}(\sigma_k \cdot e, \emptyset)) = \ok$ and
	$\pi_1(\sem{P}(\sigma_k \cdot e, \emptyset)) = \ok$, meaning that $\pred{P}(e, \emptyset) = \mtrue$. 
	By induction hypothesis applied to ${<}\phi_{H'}$, 	$\hm_0({<}\phi_{H'}) = \ok$. 
	So, by definition of $\updH$:
	\[
	\updH(\hm, e, d^*)({<}({<}\phi_{H'} \land P)) = \ok.
	\]
\end{proof}

\begin{lemma}[Historical memory completeness]\label{lem:hm-compl}
	Let $\pol$ be a DSL program, and let $(q_0, m_0) \xrightarrow{\sigma_k} (q, m)$ be the
	run of $\mathcal{A} = \GenMon(\NF(\pol))$ on $\sigma_k$.
	For every ${<}\phi_H \in \pastF(\pol)$:
	\[
	\pi_1(\sem{{<}\phi_H}(\sigma_k, \emptyset)) = \ok
	\;\implies\;
	m_\hm({<}\phi_H) = \ok.
	\]
\end{lemma}
\begin{proof}
	By induction on the length of the trace $k$.
	
	\medskip
	\noindent\emph{Base case} ($k = 0$):
	$\pi_1(\sem{{<}\phi_H}(\sigma_0, \emptyset)) = \pending \neq \ok$ since $\sigma_0$ is the empty
	trace. The hypothesis is vacuously satisfied.
	
	\medskip
	\noindent\emph{Inductive step}: assume the lemma holds for a trace $\sigma_{k-1}$ and the corresponding run $(q_0, m_0) \xrightarrow{\sigma_{k-1}} (q'', m'')$, i.e.,
	for every ${<}\phi_H \in \pastF(\pol)$:
	\[
		\pi_1(\sem{{<}\phi_H}(\sigma_{k-1}, \emptyset)) = \ok \;\implies\; m''_\hm({<}\phi_H) = \ok.
	\]
	Consider $\sigma_k$ and the last transition of the run. 
	By rule \rulename{Par}, $m_\hm = \updH(m''_\hm, e_k, d^*)$.
	We apply \Cref{lem:updh-compl} with $\hm = m''_\hm$, $\sigma_{k-1}$
	as the trace of length $k-1$, and $e = e_k$.
	
	The premise of \Cref{lem:updh-compl} requires:
	\[
	\pi_1(\sem{{<}\phi_H}(\sigma_{k-1}, \emptyset))
	= \ok
	\;\implies\;
	m''_\hm({<}\phi_H) = \ok,
	\]
	which holds because it is exactly our induction hypothesis. 
	Then, by applying \Cref{lem:updh-compl}, for every ${<}\phi_H \in \pastF(\pol)$:
	\[
	\pi_1(\sem{{<}\phi_H}(\sigma_k, \emptyset)) = \ok \;\implies\; m_\hm({<}\phi_H) = \ok.
	\]
\end{proof}

\moncompl*
\begin{proof}
	By \Cref{th:nf-sem}, $\sem{\pol}(\sigma_k) = \violated$ implies that the normal form formula is violated too: $\pi_1(\sem{\NF(\pol)}(\sigma_k, \emptyset)) = \violated$. 
	By the semantics of $G$, there exists a conjunct of $\NF(\pol)$ that is
	violated.  Let
	$i^*$ be the index of the first violating prefix of
	$\sigma_k$ (well-defined by \Cref{prop:fvp}). 
	We show that the run
	\[
	(q_0, m_0) \xrightarrow{\sigma_{i^*}}
	(q, m)
	\]
	of $\mathcal{A}$ on $\sigma_{i^*}$ reaches a bad configuration. 
	We proceed by cases on which conjunct is violated.
	
	\medskip
	\noindent\emph{Safety case.} We assume that a conjunct of the form $\lnot\psi_i$ is violated, namely $\pi_1(\sem{\lnot\psi_i}(\sigma_{i^*}, \emptyset)) = \violated$. 
	By the semantics of $\lnot$ at the formula level,
	we have $\pi_1(\sem{\psi_i}(\sigma_{i^*},\emptyset)) = \ok$.
	
	By applying \Cref{lem:hm-compl} to the run $(q_0, m_0)\xrightarrow{\sigma_{i^*}} (q, m)$, we have for all ${<}\phi_H \in \pastF(\pol)$,
	$\pi_1(\sem{{<}\phi_H}(\sigma_{i^*}, \emptyset)) = \ok$ implies 
	$m_\hm({<}\phi_H) = \ok$.
	
	By applying \Cref{lem:eval-compl} at step $i^*$, we have:
	\[
	\eval{\psi_i}{m_\hm}{e_{i^*}} = \mtrue.
	\]
	The premise of rule \rulename{S-Viol} is satisfied and $A^s_i$
	transitions to the $\violated$ state at step $i^*$, hence $(q, m) \in \bad$ with
	$i^* \leq k$. 
	
	\medskip
	\noindent\emph{Bounded-response case.} 
	We assume that the violated conjunct is of the form $P^j_1 \Fboundj P^j_2$.
	By the semantics of $\Fboundj$, there exists $(i, \rho_i, d_i) \in \mathcal{I}(\sigma_{i^*},
	\emptyset)$ with $e_{i^*}.t > d_i$ and $\lnot\sat(i, \rho_i, d_i, \sigma_{i^*})$.
	We show that the run $(q_0, m_0) \xrightarrow{\sigma_{i^*}} (q, m)$
	reaches a bad configuration, i.e., that \rulename{R-Viol} fires at step $i^*$.
	
	\medskip
	\noindent\emph{Step 1 --- instance insertion.}
	By definition of $\mathcal{I}(\sigma_{i^*},\emptyset)$, there exists $i < i^*$ such that
	$\pred{P^j_1}(e_i, \emptyset) = (\mtrue,\rho_i)$ and $d_i = e_i.t + \Delta_j$, where $e_i$ is the last event of $\sigma_i$.
	Since $i < i^*$, the run on $\sigma_{i^*}$ passes through step $i$, and this is precisely the
	condition under which $\mathsf{new}$ fires: it inserts $(d_i, \rho_i)$ into $I_j$ at step $i$.
	
	\medskip
	\noindent\emph{Step 2 --- instance preservation.}
	The instance $(d_i, \rho_i)$ remains in $I_j$ at every step $\ell \in (i, i^*]$. 
	By definition of $\mathsf{rm}$, the instance $(d_i, \rho_i)$ is removed at step $\ell$
	only in two cases. 
	When $e_\ell.t \leq d_i$ and $\pred{P^j_2}(e_\ell, \rho_i) = \mtrue$: this is precisely the
	condition defining $\sat(i, \rho_i, d_i, \sigma_{i^*})$, contradicting the hypothesis
	$\lnot\sat(i, \rho_i, d_i, \sigma_{i^*})$.
	When $e_\ell.t > d_i$: since $i^*$ is the first violating prefix, $e_{i^*}$ is the
	first event with $e_{i^*}.t > d_i$. 
	Therefore, no step $\ell \in (i, i^*)$ has $e_\ell.t > d_i$, and the only step in $(i, i^*]$ with $e_\ell.t > d_i$ is $\ell = i^*$ itself, at which \rulename{R-Viol} fires before $\mathsf{rm}$
	can remove the instance.
	In both cases the instance is not removed before step $i^*$.

	\medskip
	\noindent\emph{Step 3 --- violation detection.}
	At step $i^*$, we verify that $\mathsf{vi}(I_j, e_{i^*}$ holds. 
	By definition of $\mathsf{vi}$, this requires $(d_i, \rho_i) \in I_j$
	with $e_{i^*}.t > d_i$.
	The instance $(d_i, \rho_i)$ is in $I_j$ at step $i^*$ by Step~2,
	and $e_{i^*}.t > d_i$ by the semantics of $\Fboundj$ applied to $\sigma_{i^*}$.
	Therefore $\mathsf{vi}(I_j,	e_{i^*})$ holds, $U$ returns the flag $\violated$, and by rule \rulename{R-Viol},
	$A^r_j$ transitions to $\violated$ at step
	$i^*$. Hence $(q, m) \in \bad$
	with $i^* \leq k$. 
\end{proof}

%% file: sections/appendix-impl.tex

This appendix details the additional language features mentioned in \Cref{sec:impl:tool}: four temporal operators (\Cref{app:temporal}), two reuse mechanisms (\Cref{app:reuse}), and the standard library (\Cref{app:stdlib}).

\subsection{Temporal Operators}
\label{app:temporal}

Bounded sequences and bounded safety are desugared into the base temporal logic of \Cref{sec:lang}, so the soundness and completeness results of \Cref{sec:synthesis} apply directly.
Bounded until and bounded count require dedicated compilation schemes, described below.

\paragraph{Bounded sequences.}
A \emph{bounded sequence} clause specifies a multi-step protocol in which each step must occur within a given deadline after the previous one:
\begin{center}
\texttt{sequence $E_1$ then within $\Delta_1$ $E_2$ then within $\Delta_2$ $E_3$ \ldots}
\end{center}
The desugaring translates a sequence of $n$ steps into $n{-}1$ independent bounded-response clauses.
For each consecutive pair $(E_i, E_{i+1})$, the compiler generates a synthetic let-declaration that records whether $E_i$ has occurred and a clause
$\texttt{when}\; E_i \;\texttt{then within}\; \Delta_i\; E_{i+1}$
requiring the next step within the specified deadline.
At runtime, each clause is handled by the standard bounded-response mechanism described in \Cref{sec:processing}: it creates an entry in the response map when $E_i$ fires and raises a violation if $E_{i+1}$ does not arrive before the deadline expires.

\paragraph{Bounded until.}
A \emph{bounded until} clause asserts that once a condition $P_1$ holds, a response $P_2$ must occur within $\Delta$:
\begin{center}
\texttt{while $P_1$ then within $\Delta$ $P_2$}
\end{center}
Unlike bounded sequences, this construct is not desugared but compiled directly into a two-state \emph{extended finite-state machine} (EFSM) maintained in a dedicated BPF map, keyed by $(\kappa, \mathit{clause\_id})$.
The machine operates as follows:
\begin{itemize}[nosep]
  \item \emph{Idle.} The initial state. When an event matching $P_1$ arrives, the machine transitions to \emph{Active} and records a deadline of $e.t + \Delta$.
  \item \emph{Active.} On each subsequent event: if $P_2$ matches, the obligation is fulfilled and the machine returns to \emph{Idle}; if the deadline has expired, a violation is raised and the machine resets.
\end{itemize}
As with all bounded-response clauses, if no event arrives before the deadline, the timeout scanner in the runtime daemon (\Cref{sec:runtime}) detects the expiry.

\paragraph{Bounded count.}
A \emph{bounded count} clause raises a violation when an event $E$ occurs $N$ or more times within a sliding time window:
\begin{center}
\texttt{count $E \geq N$ within $\Delta$}
\end{center}
Note that \texttt{exists($E$, $\Delta$)} is syntactic sugar for \texttt{count $E \geq 1$ within $\Delta$}.
The construct compiles to a per-key counter stored in a dedicated BPF map, keyed by $(\kappa, \mathit{clause\_id})$.
Each entry holds a counter and a window-start timestamp.
When an event matching $E$ arrives: if the current time exceeds the window start by more than $\Delta$, the window is reset (counter set to~1, timestamp updated); otherwise, the counter is incremented.
A violation is raised whenever the counter reaches $N$.

\paragraph{Bounded safety.}
A \emph{bounded safety} clause raises a violation when a predicate $P$ holds while a
guard $G$ has been observed in the recent past:
\begin{center}
\texttt{forbid $P$ when $G$ within $\Delta$}
\end{center}
Intuitively, the clause forbids $P$ whenever the guard $G$ has been observed within the preceding $\Delta$ time units; for example, ``forbid an outbound connection to a destination that was already contacted in the last ten seconds'' (beacon detection).
Unlike \texttt{when $G$ then within $\Delta$ $P$} of \Cref{sec:lang}, which is
a \emph{liveness} obligation (a missing $P$ is the violation), bounded safety
is a \emph{safety} property: the violation is raised \emph{inline} in the BPF
hook when $P$ fires, with no need for a timeout scanner.
The construct is syntactic sugar for the conjunction of a bounded count on
$G$ over the same window and a standard \texttt{forbid $P$ when $H$} clause
guarded by the resulting historical predicate $H$; the compiler implements it
directly by emitting the same per-key counter of \emph{bounded count} and
checking it at the firing of $P$, avoiding the creation of an intermediate
history entry.

\subsection{Reuse Constructs}
\label{app:reuse}

Both constructs below are resolved entirely at desugaring time; therefore, the \thetool compilation pipeline (type checking, normalization, code generation) remains unchanged.

\paragraph{Parameterized templates.}
A \emph{parameterized template} is a named, reusable block of policy clauses with formal parameters:
\begin{lstlisting}[language=bpfence]
let rate_check(ev, threshold, window) = {
  forbid count ev(any) >= threshold within window
}
\end{lstlisting}
A template is instantiated at the use site via the \texttt{use} construct, which binds actual arguments to the formal parameters:
\begin{lstlisting}[language=bpfence]
use rate_check(connect, 100, 60s)
\end{lstlisting}
The desugaring phase resolves each \texttt{use} by looking up the template definition, validating that all formal parameters are supplied, and substituting the actual arguments into the template body by replacing event names in predicates and duration literals as appropriate.
The expansion is recursive: if a template body contains further \texttt{use} invocations, they are expanded in turn.
Cyclic template references are detected at compile time and reported as an error.

\paragraph{Policy inheritance.}
A policy can \emph{extend} another policy, inheriting all of its clauses:
\begin{lstlisting}[language=bpfence]
policy base_sandbox {
  apply to pid action deny
  on violation { log }
  forbid exec("/bin/sh")
  forbid load_module(_)
}
policy strict_sandbox extends base_sandbox {
  apply to pid action kill
  on violation { log }
  forbid connect(_, _)
}
\end{lstlisting}
The desugaring phase resolves inheritance transitively: if policy $A$ extends $B$ and $B$ extends $C$, the clauses of $C$ and $B$ are prepended to those of $A$ in order.
The child policy retains its own declared action (e.g., \Deny, \Kill); only the clause body is inherited.
Circular inheritance chains are detected via a visited set during resolution and reported as a compile-time error.

\subsection{Standard Library}
\label{app:stdlib}

\thetool ships a standard library that provides pre-defined schema modules and policy templates, allowing users to write policies against common Linux subsystems without specifying low-level hook signatures.

\paragraph{Schema modules.}
The standard library ships four schema modules covering 14 events across file access, process lifecycle, network connections, and kernel operations (\texttt{linux.files}, \texttt{linux.process}, \texttt{linux.network}, and \texttt{linux.kernel}).
Each event is mapped to a specific kernel hook and declares the fields extracted from the hook arguments.
Of the 14 events, 13 attach to LSM hooks (controllable: enabling \Deny and \Kill), and one (\texttt{clone}, via tracepoint) is only observable (\Alert).
The full list is in the online repository.

\paragraph{Policy templates.}
Four reusable policy templates are provided, each implementing a well-known security pattern:
\begin{itemize}[nosep]
  \item The \emph{Chinese Wall} template prevents an entity from accessing resources in two competing classes (parameterized by two path patterns), implementing the Brewer--Nash model~\cite{brewer1989chinese} via historical predicates.
  \item The \emph{Sandboxing} template restricts a set of dangerous operations (\texttt{exec}, network \texttt{send}, \texttt{mount}, \texttt{load\_module}) after a triggering event, parameterized by the trigger predicate.
  \item The \emph{Data-loss prevention (DLP)} template forbids network transmission to a given destination after data has been read from a sensitive path, using cross-event correlation.
  \item The \emph{Rate limiting} template raises a violation when an event exceeds a fixed threshold within a time window, implemented via the bounded-count operator.
\end{itemize}
Each template is instantiated via the \texttt{use} construct described in \Cref{app:reuse}, requiring only the policy-specific parameters (e.g., path patterns, trigger events) to produce a fully functional policy.